\documentclass[12pt]{article}

\usepackage[utf8]{inputenc}
\usepackage{lmodern}
\usepackage{tikz}
\usepackage[bottomfloats,belowfloats]{footmisc}
\usepackage{threeparttable}

\usepackage{amsmath,amssymb,amsthm}
\usepackage{dsfont}
\allowdisplaybreaks

\usepackage[letterpaper, margin=1in, includehead, includefoot]{geometry}
\usepackage{setspace}
\usepackage{changepage}
\usepackage{titling}
\usepackage[affil-it]{authblk}
\usepackage{algorithm}
\usepackage{algpseudocode}

\usepackage{enumitem}
\usepackage{booktabs}
\usepackage{xcolor}
\usepackage{graphicx}

\usepackage{xcolor}

\usepackage{natbib}
\definecolor{MyDarkBlue}{rgb}{0,0.08,0.55}
\usepackage[
    colorlinks=true,
    linkcolor=MyDarkBlue,
    citecolor=MyDarkBlue,
    urlcolor=MyDarkBlue
]{hyperref}

\theoremstyle{definition}
\newtheorem{remark}{Remark}
\newtheorem{example}{Example}

\theoremstyle{plain}
\newtheorem{theorem}{Theorem}
\newtheorem{definition}{Definition}
\newtheorem{proposition}{Proposition}

\newtheorem{lemma}{Lemma}

\newcommand{\R}{\mathbb{R}}

\newcommand{\bP}{\mathbf{P}}
\newcommand{\E}{\mathbb{E}}
\newcommand{\1}{\mathds{1}}
\newcommand{\ind}{\perp\!\!\!\!\perp}

\newcommand{\K}[0]{\mathcal{K}}

\newcommand{\F}[0]{\mathcal{F}}

\title{Partially Identified Ambiguity}

\date{First Posted: February 11, 2026 \\ 
This Version: November 19, 2026}

\author{Cheaheon Lim\thanks{Department of Economics, Harvard University.  I am indebted to Tomasz Strzalecki,  Elie Tamer, and Davide Viviano for detailed comments that significantly improved earlier versions of this draft. I also thank Isaiah Andrews, Aislinn Bohren, Xiaoyu Cheng, Tommaso Denti, Francesco Fabbri, Brian Hill, Peter Klibanoff, Shengwu Li, José Luis Montiel Olea, Yechan Park, Matthew Rabin, Andrei Savochkin, Marciano Siniscalchi,  Vitalii Tubdenov,  Weijie Zhong, and all workshop and seminar participants at Harvard University, RUD 2026, and NU-Kellogg Robustness Summer School for their insightful comments and suggestions.}}

\begin{document}

\maketitle

\onehalfspacing
 
\begin{abstract}
\begin{adjustwidth}{1.5cm}{1.5cm}
\onehalfspacing
This paper addresses the normative question of whether one should condition on the data when performing robust Bayesian analysis under partial identification. We provide an ex ante representation of the conditional Gamma-minimax criterion and show that it can violate Aumann plausibility, a set-valued generalization of the rational expectations property. Our main characterization result establishes that Aumann plausibility holds if and only if the decision maker's prior set exhibits partially identified ambiguity. The latter condition requires the prior set to equal the Aumann expectation of a random set of distributions supported on the parameter's identified sets. We interpret these results as suggesting that conditional Gamma-minimax should only be used under partially identified ambiguity. The implications for various robust Bayesian methods are discussed. 
\end{adjustwidth}
\end{abstract}


\newpage

\section{Introduction}

The data we encounter in the world seldom allow us to make precise inferences. Missing or mismeasured data, imperfect instruments, and limited external validity can all prevent point identification of the parameter of interest \citep[e.g.,][]{kline2023recent}. A growing literature draws on the statistical decision theory of \cite{Wald_1950_SDT} to study optimal decision making under partial identification, with applications to treatment choice, evidence aggregation, and optimal pricing.\footnote{See, e.g.,  \cite{Stoye_2012_AR, yata2021optimal, ishihara2021evidence, manski2025identification, christensen2026optimal, fernandez2026robust, montiel2026decision}.} In such settings, standard Bayesian analysis can be highly sensitive to the choice of prior because the data does not resolve uncertainty about unidentified features of the model \citep[e.g.,][]{Moon_Schorfheide_ECTA_2012}. We focus  on the robust Bayesian approach, which addresses this issue by representing the ambiguity faced by the decision maker (DM) over the identified set of parameters using a set of priors \citep[e.g.,][]{Giacomini_Kitagawa_Read_2025}.

The use of multiple priors introduces a fundamental inconsistency that does not arise in standard Bayesian decision theory. An ex ante optimal decision rule may prescribe actions that are no longer optimal once beliefs are updated after the data realization \citep[e.g.,][]{vidakovic2000gamma, Stoye_2012_AR}.  This discrepancy gives rise to  two different robust Bayesian decision criteria: ex ante and conditional $\Gamma$-minimax. Given some set of prior beliefs,  the former selects the decision rule that minimizes worst-case ex ante expected loss, while the latter selects the action that minimizes worst-case conditional expected loss at each data realization.  Although the distinction between these criteria is well-known, there is little consensus on when one criterion should be used over another.\footnote{For instance, \cite{chamberlain2000econometric, armstrong2022robust} adopt the ex ante criterion, while \cite{giacomini2019estimation, christensen2026optimal, vilfort2026headline} use the conditional criterion. \cite{vidakovic2000gamma, Stoye_2012_AR, Giacomini_Kitagawa_Read_2025, fernandez2026robust}  compare the two criteria across different applications, but they do not provide a general basis for choosing between them.} This paper addresses this gap and provides normative guidance on the choice between the two criteria based on the DM's set of priors and specified statistical model.

We first provide an ex ante representation of the conditional $\Gamma$-minimax criterion. This allows us to directly compare the two forms of robust Bayes decision making through the set of joint distributions over the parameter and data space that rationalize their decisions.  Ex ante $\Gamma$-minimax evaluates decision rules under joint distributions obtained by combining each prior in the DM's prior set with the specified statistical model. By contrast, conditional $\Gamma$-minimax can be represented as evaluating decision rules against an enlarged set of joint distributions. This provides a belief-based interpretation of why conditional $\Gamma$-minimax is often viewed as the more ``conservative" criterion \citep[e.g.,][]{christensen2026optimal}.  However, our representation result (Proposition~\ref{prop:gamma_characterization}) suggests caution in applying this interpretation, as the enlarged set of joint distributions can induce marginal beliefs about the parameter that lie outside the DM's original prior set.  We argue that this inconsistency is problematic because  the DM's subjective ambiguity, as encoded by her prior set, should be preserved regardless of the choice of optimality criterion.

Suppose each prior in the prior set induces a common predictive distribution over the data.\footnote{This condition reflects the intuition that ambiguity under partial identification arises because parameter values within the same identified set induce the same data likelihood \citep[e.g.,][]{Giacomini_Kitagawa_ECTA_2021}.} The requirement that conditional $\Gamma$-minimax preserve the DM's subjective ambiguity then admits an intuitive formulation as the set-valued generalization of the rational expectations property.\footnote{The rational expectations property is also referred to as Bayes plausibility or the martingale property of beliefs \citep[e.g.,][]{KG_bayesian_persuasion, bergemann2016bayes, Augenblick_Rabin_2021_QJE}.} The set of posterior distributions obtained by updating each prior in the DM's prior set can be be viewed as a random set in the space of probability distributions over the parameter, with randomness governed by the common prior predictive. If the Aumann expectation of this random set coincides with the DM's prior set, then the joint beliefs underlying conditional $\Gamma$-minimax  preserve the DM's original prior ambiguity. This condition, which we refer to as \emph{Aumann plausibility},  reflects the rational expectations property by requiring the DM's ex ante valuation of actions to coincide with the expectation of her conditional valuation after observing the data.


Our main characterization result shows that Aumann plausibility holds if and only if the DM's prior set reflects \emph{partially identified ambiguity} (Theorem~\ref{thm:two_ways}).  Formally, we say that a prior set is partially identified if it equals the Aumann expectation of a random set of probability distributions supported on the identified sets of the parameter. The ``full ambiguity" prior set of \cite{Giacomini_Kitagawa_ECTA_2021} is a special case, where the random set maps to the collection of all probability distributions supported on each identified set.  We illustrate how  common approaches to robust Bayesian analysis, such as the KL-neighborhood and $\epsilon$-contamination specifications \citep[e.g.,][]{giacomini2019estimation, giacomini2022uncertain}, can be accommodated in our setting. Our framework also allows us to extend Blackwell's theorem of informativeness to conditional $\Gamma$-minimax decision making under partially identified ambiguity (Proposition~\ref{prop:Blackwell}).

We view ex ante $\Gamma$-minimax as the benchmark criterion for robust Bayesian analysis because it preserves both the DM's subjective prior set and statistical model. Nevertheless, conditional $\Gamma$-minimax is attractive because it avoids the dynamic inconsistency of ex ante optimal rules \citep[e.g.,][]{lim2026dynamically} and is generally more computationally and analytically tractable \citep[e.g.,][]{giacomini2019estimation, christensen2026optimal}. Our results clarify when the use of conditional $\Gamma$-minimax can be justified as the ``more conservative" approach relative to ex ante $\Gamma$-minimax. In particular, Aumann plausibility ensures that the criterion does not implicitly introduce beliefs outside the DM's original prior set. Since partially identified ambiguity is necessary and sufficient for Aumann plausibility, we argue that conditional $\Gamma$-minimax should only be used under partially identified ambiguity.


\subsection{Related Literature} 

This paper contributes to the literature on statistical decision making under partial identification \citep[e.g.,][]{Stoye_2012_AR,  d2021orthogonal, yata2021optimal, ishihara2021evidence, kido2022distributionally, manski2025identification}, with an emphasis on robust Bayesian methods \citep[e.g.,][]{Giacomini_Kitagawa_ECTA_2021, Giacomini_Kitagawa_Read_2025, fernandez2026robust, christensen2026optimal}. In particular, we add to the discussion on the dynamic inconsistency of ex ante and conditional $\Gamma$-minimax decisions \citep[e.g.,][]{vidakovic2000gamma, kitagawa2012gamma_minimax, Stoye_2012_AR, fernandez2026robust, lim2026dynamically} by identifying when conditional decision making implicitly introduces beliefs that depart from the DM's subjective ambiguity.  

We also contribute to the literature that employs random set theory to study inference under partial identification \citep[e.g.,][]{beresteanu2008asymptotic, beresteanu2011sharp, beresteanu2012partial, molchanov2014applications, molinari2020microeconometrics, Giacomini_Kitagawa_ECTA_2021}. Our results illustrate how the language of random set theory allows us to generalize properties of standard Bayesian decision theory to robust Bayesian analysis. For instance, our definition of partially identified ambiguity relies on expressing the DM's ambiguity set as the Aumann expectation of a random set of priors, and our Aumann plausibility condition extends the rational expectations property of Bayesian beliefs to the robust Bayesian setting.  


From the perspective of axiomatic decision theory, whether the DM has beliefs compatible with partially identified ambiguity can be determined from preferences over state-contingent outcomes.  Partially identified ambiguity is a special case of the ambiguity that arises from generalizing subjective expected utility  theory to  ``smaller" domains, the axiomatic foundations of which have been studied extensively \citep[e.g.,][]{Epstein_Zhang_ECTA_2001, Zhang_ET_2002, Kopylov_JET_2007, Lehrer_AEJ_2012, GP_EUU, Denti_Pomatto_WP_2020}.  Our ex ante representation of the conditional $\Gamma$-minimax criterion also relates closely to the notion of rectangularity introduced by \cite{ES_recursive}, which characterizes dynamic consistency under the multiple priors model of \cite{GS_MEU}.  


The remainder of the paper is organized as follows. Section~\ref{sec:example} presents a motivating example and  Section~\ref{sec:model} introduces the formal model. Section~\ref{sec:main-results} develops the ex ante representation of the conditional $\Gamma$-minimax criterion and establishes our main characterization result for Aumann plausibility and partially identified ambiguity. Section~\ref{sec:discussion} discusses the normative implications and applications of our results, and Section~\ref{sec:conclusion} concludes.

\section{Motivating Example} \label{sec:example}

\subsection{Preliminaries}

To illustrate the distinction between ex ante and conditional $\Gamma$-minimax analysis, we first present a simple example motivated by the evidence aggregation setting of \cite{ishihara2021evidence}. Let $\theta$ and $\phi$ denote the treatment effect of a new policy in the target and study populations, respectively. The DM is deciding whether to approve the policy in the target population, so  $\theta$ is the only payoff-relevant (i.e., structural) parameter.  Due to limited external validity, the two effects need not coincide, but we know that $\theta$ lies in the identified set $IS_\Theta(\phi) = [\phi-b, \phi+b]$ for some known constant $b>0$. 

For simplicity, we assume throughout this section that $\Phi=\{-2,2\}$ and $b=1$, such that  $\Theta = [-3,-1] \cup [1,3]$ with  $IS_\Theta(-2) = [-3,-1]$ and $IS_\Theta(2) = [1,3]$.  Suppose the  DM observes binary data $Y \in \{y_1, y_2 \}$ about the study population, with likelihood
\begin{align*}
    \ell(y_1 \mid \phi)  =  \begin{cases}
       p, & \phi = -2,\\
       1-p, & \phi = 2,
    \end{cases}
    \qquad  \frac{1}{2} \leq p \leq 1.
\end{align*}
Because the data does not distinguish between different parameters in each $IS_\Theta(\phi)$, we write $\ell(y \mid \theta)$ to denote $\ell(y \mid \phi)$ whenever $\theta \in IS_\Theta(\phi)$. The DM's objective is to choose a decision rule $\delta : \{y_1, y_2 \} \to \{0,1\}$ that rejects or approves the policy based on the observed data.

\paragraph{Ambiguity Sets} We consider two specifications of the DM's subjective ambiguity over the parameter space $\Theta$, represented by the following sets of priors:
\begin{align*}
   \Pi_0^F  & =  \left\{ \frac{1}{2} \pi_- + \frac{1}{2} \pi_+ : \pi_-\in\Delta([-3,-1]), \,  \pi_+ \in \Delta([1,3])   \right\}   \\ 
   \Pi_0^R & = \left\{ \pi \in \Pi_0^F : \E_{\pi}[\theta] = 0 \right\},
\end{align*}
where $\Delta(\cdot)$ denotes the collection of probability measures over the corresponding space.

The first set $\Pi_0^F$ corresponds to the ``full ambiguity" specification commonly used in the literature \citep[e.g.,][]{Giacomini_Kitagawa_ECTA_2021, christensen2026optimal}. Under this specification, every prior in $\Pi_0^F$ assigns equal probability to the identified sets at $\phi \in \{-2,2\}$, but the DM faces complete ambiguity about the conditional distribution of the target population treatment effect  given $\phi$.  The second set $\Pi_0^R$  imposes the additional moment restriction that the expected treatment effect in the target population is zero. Since the DM's uniform prior over $\Phi$ implies $\E[\phi]=0$, this restriction can also be interpreted as the belief that the external validity bias $\theta-\phi$ is zero on average. 


\paragraph{Robust Bayesian Criteria}  Given an ambiguity set $\Pi_0 \in \{\Pi_0^F, \Pi_0^R\}$, there are two approaches to robust Bayesian analysis. The ex ante $\Gamma$-minimax criterion evaluates decision rules using the worst-case belief over the set 
\begin{align*}
    \Gamma  = \left\{ \pi_0(d\theta)\ell(dy \mid \theta) : \pi_0 \in \Pi_0 \right\},
\end{align*}
which consists of all joint distributions  induced by the priors in $\Pi_0$ and the likelihood $\ell$. 

By contrast, the  conditional $\Gamma$-minimax criterion evaluates decision rules separately at each data realization $y \in Y$ using the posterior set $\Pi_y$ obtained by updating each prior in $\Pi_0$ under $\ell$. We show in Section~\ref{sec:main-results} that the  conditional $\Gamma$-minimax criterion admits an equivalent ex ante representation based on the joint belief set
\begin{align*}
    \Gamma^* &  = \left\{ \pi_y(d\theta) \mu(dy) : \pi_y \in \Pi_y  \, \text{ for all } \,  y \in Y \right\},
\end{align*}
where $\mu$ is the prior predictive distribution induced by the likelihood $\ell$ and any $\pi_0 \in \Pi_0$. In our example, we have $\mu(y_i) \equiv \int_\Theta \ell(y_i \mid \theta) \pi_0(d\theta) = \frac{1}{2}$ for all $\pi_0 \in \Pi_0^F \supseteq \Pi_0^R$ and $i \in \{1,2\}$.\footnote{In general, $\Gamma^* \supseteq \Gamma$ because $\Gamma$ requires posterior beliefs across data realizations  to arise from a common prior $\pi_0 \in \Pi_0$, while $\Gamma^*$ allows them to arise from different priors. Formally, $\Gamma^*$ is the \emph{rectangular hull} of $\Gamma$ with respect to the filtration generated by data realizations (see Section~\ref{sec:main-results} and Appendix~\ref{Appendix:rectangularity} for details).}

\subsection{Full Ambiguity}

We first illustrate the distinction between the two robust Bayes criteria by comparing the two sets of joint distributions induced by the full ambiguity set $\Pi_0^F$. We consider the projections of $\Gamma$ and $\Gamma^*$ onto the space of conditional means $(\bar \theta_1^\gamma, \bar \theta_2^\gamma) \equiv (\E_\gamma[\theta \mid Y=y_1], \E_\gamma[\theta \mid Y=y_2])$, where either $\gamma \in \Gamma$ or $\Gamma^*$.  Under $\Gamma$, Bayes rule implies that
\begin{align*}
    \left\{ (\bar \theta_1^\gamma, \bar  \theta_2^\gamma):\gamma\in\Gamma\right\}  =   \Big\{
    \big(  p t_- + (1-p) t_+,  \,  (1-p)  t_- + p t_+   \big): t_-\in[-3,-1],\, t_+ \in [1,3]  \Big\},
\end{align*}
where $t_-, t_+$ are the within-identified-set means over $IS_\Theta(-2)$ and $IS_\Theta(2)$, respectively.   By contrast, $\Gamma^*$ allows posterior beliefs to be selected separately across data realizations,  and the projection of $\Gamma^*$ is the Cartesian product
\begin{align*}
    \left\{(\bar \theta_1^\gamma, \bar \theta_2^\gamma):\gamma \in \Gamma^*\right\}  =  [1-4p,3-4p]\times[-3+4p,-1+4p].
\end{align*}

\begin{figure}[t]
\centering
\begin{tikzpicture}[
    x=0.8cm, y=0.8cm,
    line join=round, line cap=round
]
\newcommand{\meanpanel}[3]{%
  \begin{scope}[shift={(#1,0)}]
    \pgfmathsetmacro{\pp}{#2}
    \pgfmathsetmacro{\xlo}{1-4*\pp}
    \pgfmathsetmacro{\xmid}{2-4*\pp}
    \pgfmathsetmacro{\xhi}{3-4*\pp}
    \pgfmathsetmacro{\ylo}{-3+4*\pp}
    \pgfmathsetmacro{\ymid}{-2+4*\pp}
    \pgfmathsetmacro{\yhi}{-1+4*\pp}

    \filldraw[
        fill=blue!12,
        draw=blue!65!black,
        thick
    ]
        (0,0) rectangle (4,4);

    \filldraw[
        fill=orange!28,
        draw=orange!85!black,
        very thick
    ]
        (0,0)
        -- ({4*\pp},{4*(1-\pp)})
        -- (4,4)
        -- ({4*(1-\pp)},{4*\pp})
        -- cycle;

    \draw[->,black!65] (0,0)--(4.35,0);
    \draw[->,black!65] (0,0)--(0,4.35);

    \foreach \t in {0,2,4} {
        \draw[black!60] (\t,0.06)--(\t,-0.06);
        \draw[black!60] (0.06,\t)--(-0.06,\t);
    }

    \node[below=3pt,font=\scriptsize] at (0,0)
        {$\pgfmathprintnumber{\xlo}$};
    \node[below=3pt,font=\scriptsize] at (2,0)
        {$\pgfmathprintnumber{\xmid}$};
    \node[below=3pt,font=\scriptsize] at (4,0)
        {$\pgfmathprintnumber{\xhi}$};

    \node[left=3pt,font=\scriptsize] at (0,0)
        {$\pgfmathprintnumber{\ylo}$};
    \node[left=3pt,font=\scriptsize] at (0,2)
        {$\pgfmathprintnumber{\ymid}$};
    \node[left=3pt,font=\scriptsize] at (0,4)
        {$\pgfmathprintnumber{\yhi}$};

    \node[font=\scriptsize] at (2,-0.9)
        {$\bar\theta_1^\gamma$};
    \node[font=\scriptsize,rotate=90] at (-1.1,2)
        {$\bar\theta_2^\gamma$};

    \node[font=\small] at (2,4.65) {#3};

  \end{scope}%
}

\meanpanel{0}{0.5}{$p=1/2$}
\meanpanel{5.9}{0.75}{$p=3/4$}
\meanpanel{11.8}{0.9}{$p=9/10$}

\end{tikzpicture}
\caption{Projections of $\Gamma$ (orange) and $\Gamma^*$ (blue) onto the conditional mean space $(\bar\theta_1^\gamma,\bar\theta_2^\gamma)$ under the full ambiguity set $\Pi_0^F$.}
\label{fig:project_1}
\end{figure}

Figure~\ref{fig:project_1} illustrates the difference between the projections of $\Gamma$ and $\Gamma^*$. When $p = \frac{1}{2}$, the data is uninformative, so $\Gamma$ requires the posterior means to coincide across data realizations, while $\Gamma^*$ allows them to differ. As $p$ increases, the discrepancy between the two projections diminishes, disappearing at $p=1$. The fact that $\Gamma^* \supseteq \Gamma$ can be interpreted in terms of the game underlying minimax analysis.  Under ex ante $\Gamma$-minimax, Nature must commit to a single pair of within-identified-set means $(t_-,t_+)$ across data realizations. By contrast, conditional $\Gamma$-minimax allows Nature to select these means separately at each data realization, yielding the Cartesian product structure of the projection of $\Gamma^*$.   At $p=1$, each $\phi$ generates only one possible data realization, so the two criteria coincide.

\subsection{Restricted Ambiguity}

Another projection of interest is the marginalization of $\Gamma$ and $\Gamma^*$ onto the space $\Delta(\Theta)$ of beliefs over the payoff-relevant parameter. We argue that these marginals should coincide with the DM's original prior set because $\Pi_0$ represents her subjective ambiguity over $\Theta$. This holds by construction for $\Gamma$ and, as we show in Section~\ref{sec:main-results}, also holds for $\Gamma^*$ under full ambiguity. More generally, the marginalization of $\Gamma^*$ is the weighted Minkowski sum
\begin{align}
    \Pi_0^* & := \left\{
    \sum_{i=1}^2 \gamma(\cdot,y_i) : \gamma \in \Gamma^*
    \right\}  = \sum_{i=1}^2 \mu(y_i)\Pi_{y_i},   \label{eq:ex-minkowski}
\end{align}
and we show below that $\Pi_0^*$ can contain beliefs that are not in $\Pi_0$ when $\Pi_0 = \Pi_0^R$.

Figure~\ref{fig:motivating-prior-projection} compares $\Pi_0^R$ and $\Pi_0^*$ through their projections onto the space of within-identified-set means. The mean-zero restriction defining $\Pi_0^R$ implies that its projection is the line segment $\left\{(-t,t): t \in [1,3] \right\}$.  Under $\Pi_0^*$, the Minkowski sum in \eqref{eq:ex-minkowski} permits combining posterior beliefs $\pi^1 \in \Pi_{y_1}$ and $\pi^2 \in \Pi_{y_2}$ that arise from different priors in $\Pi_0^R$, and the resulting marginal $\frac{1}{2} \pi^1 + \frac{1}{2}\pi^2 \in \Pi_0^*$ does not necessarily satisfy the mean-zero restriction. For each $i=1,2$, let $\pi_0^i = \frac{1}{2}\pi_-^i + \frac{1}{2}\pi_+^i \in \Pi_0$ denote the prior generating the posterior $\pi^i$. By Bayes rule, the projection of $\Pi_0^*$ onto the space of within-identified-set means is 
\begin{align*}
    \left\{ (\E_{\pi_-}[\theta], \E_{\pi_+}[\theta]) : \pi \in  \Pi_0^* \right\} =  \left\{ \big(-p t^1 - (1-p) t^2, \,   (1-p) t^1 + p t^2 \big)  : t^1,t^2 \in [1,3] \right\}.
\end{align*}

\begin{figure}[t]
\centering
\begin{tikzpicture}[
    x=0.80cm, y=0.80cm,
    line join=round, line cap=round
]
\newcommand{\priorpanel}[3]{%
  \begin{scope}[shift={(#1,0)}]
    \pgfmathsetmacro{\pp}{#2}

    \filldraw[
        fill=blue!12,
        draw=blue!65!black,
        thick
    ]
        (0,4)
        -- ({4*(1-\pp)},{4*(1-\pp)})
        -- (4,0)
        -- ({4*\pp},{4*\pp})
        -- cycle;

    \draw[
        black!45,
        densely dashed,
        thick
    ]
        (0,0) rectangle (4,4);

    \draw[
        orange!85!black,
        very thick
    ]
        (0,4)--(4,0);

    \draw[->,black!65] (0,0)--(4.35,0);
    \draw[->,black!65] (0,0)--(0,4.35);

    \foreach \t in {0,2,4} {
        \draw[black!60] (\t,0.06)--(\t,-0.06);
        \draw[black!60] (0.06,\t)--(-0.06,\t);
    }

    \node[below=6pt,font=\scriptsize] at (0,0) {$-3$};
    \node[below=6pt,font=\scriptsize] at (2,0) {$-2$};
    \node[below=6pt,font=\scriptsize] at (4,0) {$-1$};

    \node[left=3pt,font=\scriptsize] at (0,0) {$1$};
    \node[left=3pt,font=\scriptsize] at (0,2) {$2$};
    \node[left=3pt,font=\scriptsize] at (0,4) {$3$};

\node[font=\scriptsize] at (2,-1)
    {$\E_{\pi_-}[\theta]$};
\node[font=\scriptsize,rotate=90] at (-0.95,2)
    {$\E_{\pi_+}[\theta]$};
    \node[font=\small] at (2,4.65) {#3};

  \end{scope}%
}

\priorpanel{0}{0.5}{$p=1/2$}
\priorpanel{5.9}{0.75}{$p=3/4$}
\priorpanel{11.8}{0.9}{$p=9/10$}

\end{tikzpicture}
\caption{Projections of the restricted ambiguity set $\Pi_0^R$ (orange),  marginal prior set $\Pi_0^*$ (blue), and full ambiguity set $\Pi_0^F$ (dashed) on the space of within-set means $(\E_{\pi_-}[\theta],\E_{\pi_+}[\theta])$.}
\label{fig:motivating-prior-projection}
\end{figure}

Observe that $\Pi_0^*$ expands beyond $\Pi_0^R$ as the experiment becomes more informative. When $p= \frac{1}{2}$, the data is uninformative and $\Pi_0^* = \Pi_0^R$. As $p$ increases, $\Pi_0^*$ introduces priors that violate the original mean-zero restriction. At $p=1$, the implied prior set coincides with the full ambiguity set $\Pi_0^F$.  This contrasts with the previous subsection, where the difference between the projections of $\Gamma$ and $\Gamma^*$ diminished as $p$ increases. The following remark illustrates that the dilation effect we observe leads to a violation of the rational expectations property. 

\begin{remark}[Rational Expectations] \label{remark:RE-example}
    Suppose the DM has loss $L(\theta, a)= -a\theta$ and consider the always-approve rule $\delta(y) = 1$.  The $\Gamma$-minimax and $\Gamma^*$-minimax loss of $\delta$ are, respectively,
    \begin{align*}
        \sup_{\gamma \in \Gamma} \E_\gamma[-\theta]  & = \sup_{\pi_0 \in \Pi_0^R} \E_{\pi_0}[-\theta] = 0 \\
        \sup_{\gamma^* \in \Gamma^*} \E_{\gamma^*}[-\theta] & = \frac{1}{2} \sup_{\pi_1 \in \Pi_{y_1}} \E_{\pi_1}[-\theta] + \frac{1}{2} \sup_{\pi_2 \in \Pi_{y_2}} \E_{\pi_2}[-\theta]  = 2p -1.
    \end{align*}
    The ex ante $\Gamma$-minimax loss coincides with the worst-case expected loss under the DM's original prior set $\Pi_0^R$. By contrast, expected conditional $\Gamma$-minimax loss is strictly positive for $p>1/2$ and increasing in $p$. This is a stark violation of the rational expectations property, which requires the DM's ex ante valuation of actions to coincide with the expectation of her conditional valuation, regardless of the data's informativeness  (see also Section~\ref{sec:normative}).

   Suppose the DM is deciding between the always-approve and always-reject rules. The latter yields zero loss regardless of $\theta$. Assuming  ties are broken in favor of approval, the DM approves the policy under both ex ante $\Gamma$-minimax and worst-case expected loss over  $\Pi_0^R$, while she rejects the policy under conditional $\Gamma$-minimax when  $p > \frac{1}{2}$.  The discrepancy arises because the marginalization $\Pi_0^* = \frac{1}{2} \Pi_{y_1} + \frac{1}{2} \Pi_{y_2}$ of the joint beliefs underlying conditional $\Gamma$-minimax differs from the DM's original prior set $\Pi_0^R$. Since $\Pi_0^*$ is the $\mu$-Aumann expectation of the posterior sets, the requirement $\Pi_0^* = \Pi_0^R$ can be viewed as a natural set-valued generalization of the rational expectations property. We show in Section~\ref{sec:main-results} that this equality holds if and only if the DM's prior set exhibits partially identified ambiguity.
\end{remark}

\section{Model} \label{sec:model} 

\subsection{Preliminaries}

We now introduce the formal model necessary to study robust Bayesian analysis in generality.  Let $\Theta \subseteq \mathbb{R}^d$ and $Y \subseteq \mathbb{R}^m$ be Borel sets equipped with their respective Borel $\sigma$-algebras $\mathcal{B}_\Theta$ and $\mathcal{Y}$, with $d,m < \infty$. With an  abuse of notation, we also use $Y$ to denote the random variable taking values in the data space $Y$.  We write $\Delta(\cdot)$ for the collection of Borel probability measures over the corresponding space, endowed with the weak topology. Beliefs over the parameter space are denoted $\pi \in \Delta(\Theta)$.   Assume that for every $\theta \in \Theta$, the conditional distribution of $Y$ given $\theta$ admits a density $\ell(y \mid \theta)$ with respect to some $\sigma$-finite dominating measure, and that $\theta \mapsto \ell(\cdot \mid \theta)$ is weakly continuous.

In the classical statistical decision theory of \cite{Wald_1950_SDT}, the decision maker chooses a decision rule $\delta: Y \to \mathcal{A}$, where $\mathcal{A}$ is some finite action space.\footnote{We assume finiteness for expositional simplicity. All arguments below extend to compact metric spaces, such as $\Delta(\mathcal{A})$, under standard regularity conditions.} We write $\mathcal{D}$  to denote the collection of all measurable decision rules. Given some loss $L: \Theta \times \mathcal{A} \to \mathbb{R}$ that is bounded and continuous at each $a \in \mathcal{A}$, the expected loss of a decision rule $\delta$ at $\theta$ under the likelihood $\ell$ is the familiar (frequentist) risk
\begin{align*}
    R(\theta,\delta) :=
    \int_Y L(\theta,\delta(y))\,\ell(y\mid\theta)\,dy.
\end{align*}
We can then associate each decision rule $\delta$ with its risk vector $\{R(\theta,\delta)\}_{\theta\in\Theta}\in\mathbb{R}^\Theta$. Given a prior $\pi_0 \in \Delta(\Theta)$, the Bayes risk of a decision rule aggregates the risk vector by taking its expectation under $\pi_0$, such that
\begin{align*}
    r(\pi_0,\delta) :=   \E_{\pi_0}\left[R(\theta,\delta)\right]  = \int_\Theta R(\theta,\delta)\,\pi_0 (d\theta).
\end{align*}

Now consider the Bayesian's decision problem conditional on the realization of data. Let $\pi_{y} \in \Delta(\Theta)$ denote the posterior induced by $(\pi_0,\ell)$ following the realization $y \in Y$. For each action $a \in \mathcal{A}$, its posterior expected loss is
\begin{align*}
    r_y(\pi_0 ,a)  := \E_{\pi_{ y}}\left[L(\theta,a)\right]  =  \int_\Theta L(\theta,a)\,\pi_{ y}(d\theta).  
\end{align*}
The decision rule $\delta$ can then be associated with the vector $\{ r_y(\pi_0,\delta(y)) \}_{y\in Y} \in \R^Y$ of posterior losses.   To aggregate the posterior losses  across data realizations, let  $\mu  \in \Delta(Y)$ be the prior predictive distribution induced by $(\pi_0 ,\ell)$.  A standard application of Fubini's theorem and Bayes rule then yields
\begin{align*}
    \E_{\mu}  \left[ r_Y(\pi_0,\delta(Y))  \right]   & = \E_{\pi_0} \left[ R(\theta, \delta) \right],
\end{align*}
which is precisely the Bayes risk $r(\pi_0,\delta)$.

The preceding result highlights the well-known dynamic consistency of Bayesian decision making. A decision rule minimizes Bayes risk over $\mathcal{D}$ if and only if it minimizes posterior expected loss pointwise at $\mu$-almost every $y \in Y$.  Formally, this is because
\begin{align}
    \inf_{\delta: Y \to \mathcal{A}} r(\pi_0, \delta)  = \inf_{\delta: Y \to \mathcal{A}} \E_{\mu}  \left[ r_Y(\pi_0,\delta(Y))  \right]  = \E_\mu \left[ \inf_{a \in \mathcal{A}} r_Y(\pi_0, a) \right]. \label{eq:Bayes-DC}
\end{align}
Note that this equivalence is also of computational importance when solving for optimal decision rules, as  optimizing over actions at  each realization of $y$ is typically simpler than optimizing directly over functions from $Y$ to $\mathcal{A}$.

\subsection{Robust Bayesian Analysis}

The robust Bayesian approach generalizes both Bayesian and worst-case frequentist analysis by allowing the DM to specify a \emph{set} of beliefs over the parameter space.  However, unlike in the Bayesian setting, decision rules that attain the robust Bayesian analog of minimum risk do not generally coincide with the actions that minimize the robust Bayesian analog of posterior expected loss \citep[e.g.,][]{vidakovic2000gamma, Stoye_2012_AR, Giacomini_Kitagawa_Read_2025}. 

The robust generalization of the ex ante criterion induced by Bayes risk is the ex ante $\Gamma$-minimax criterion, which evaluates a decision rule according to its worst-case Bayes risk over a set of priors. Let $\mathcal{K}(\cdot)$ denote the collection of closed subsets of the corresponding space, and let $\mathcal{K}^c(\cdot)$ denote the collection of closed and convex subsets. The ex ante $\Gamma$-minimax risk at the ambiguity set  $\Pi_0 \in \K^c(\Delta \Theta)$ is
\begin{align}
     \sup_{\pi \in \Pi_0} r(\pi, \delta) & = \sup_{\gamma \in \Gamma(\Pi_0)} \E_\gamma\left[L(\theta, \delta(Y)) \right],  \label{eq:gamma_MML}
\end{align}
where $\Gamma(\Pi_0) \in \K^c(\Delta(\Theta \times Y))$ is the set of joint distributions 
\begin{align}
    \Gamma(\Pi_0)  &  :=  \left\{  \gamma(d\theta, dy) = \pi_0(d\theta) \ell(dy \mid \theta) :  \pi_0 \in \Pi_0 \right\}.  \label{eq:gamma_joint}
\end{align}
When $\Pi_0$ is a singleton, ex ante $\Gamma$-minimax risk reduces to Bayes risk. When $\Pi_0$ is the full simplex $\Delta(\Theta)$, ex ante $\Gamma$-minimax reduces to frequentist minimax loss. Note that under continuous loss, restricting attention to closed and convex prior sets is without loss  because replacing either $\Pi_0$ or $\Gamma(\Pi_0)$ with its closed convex hull leaves the suprema in \eqref{eq:gamma_MML} unchanged.

The robust analogue of the posterior expected loss criterion is conditional $\Gamma$-minimax. It evaluates each action according to its worst-case expected loss over the set of posteriors $\Pi_y \in \K^c(\Delta \Theta)$ obtained by updating each prior in $\Pi_0$. The conditional $\Gamma$-minimax loss of an action $a$ after the realization $y \in Y$ is
\begin{align}
    \sup_{\pi \in \Pi_0} r_y(\pi,a)   =  \sup_{\pi_{ y} \in \Pi_{y}}
   \E_{\pi_y}\left[ L(\theta,a) \right]. \label{eq:c_gamma_MML}
\end{align}
As we formally discuss in Section~\ref{sec:ex-ante}, decision rules that minimize ex ante $\Gamma$-minimax risk \eqref{eq:gamma_MML} do  not generally prescribe the optimal action according to \eqref{eq:c_gamma_MML} at each data realization $y \in Y$. This means that the dynamic consistency of Bayesian decision making does not generally extend to robust Bayesian analysis, making the choice between ex ante and conditional $\Gamma$-minimax criteria consequential.

\subsection{Partial Identification}  

Robust Bayesian analysis is often employed  as a tool for decision making under set-identified models \citep[e.g.,][]{Stoye_2012_AR, christensen2026optimal, fernandez2026robust, Giacomini_Kitagawa_Read_2025}. As in the motivating example, the DM's subjective ambiguity arises because the data does not point identify the payoff-relevant parameter $\theta$, and she allows for multiple conditional distributions of $\theta$ within each identified set.

Formally, set identification of $\theta$ arises because distinct parameters in $\Theta$ induce the same distribution over $Y$.  Given some likelihood $\ell: \Theta \to \Delta(Y)$, we say that $\ell$ set-identifies $\theta$ up to $\Phi$ if there exists a Borel space $(\Phi, \mathcal{B}_\Phi)$ of identified parameters and a surjective measurable map $\varphi:(\Theta, \mathcal{B}_\Theta) \to (\Phi, \mathcal{B}_\Phi)$  satisfying 
\begin{align*}
    \varphi(\theta)=\varphi(\theta')
    \quad \Leftrightarrow \quad \ell(\cdot\mid\theta)=\ell(\cdot\mid\theta').
\end{align*}
The literature often refers to $\theta \in \Theta$ as the structural (i.e., payoff-relevant) parameter and $\phi \in \Phi$ as the reduced-form parameter. For each $\phi \in \Phi$, define the identified set
\begin{align*}
    IS_\Theta(\phi)  := \left\{ \theta \in \Theta:  \varphi(\theta)=\phi \right\}.
\end{align*}
Note that $IS_\Theta(\phi) \cap IS_\Theta(\phi') = \emptyset$ whenever $\phi\neq\phi'$, and $\{IS_\Theta(\phi)\}_{\phi\in\Phi}$ partitions $\Theta$.  We  write $\theta \in \phi$ to denote $\varphi(\theta) = \phi$ whenever it is clear from context that $\phi$ refers to the identified set $IS_\Theta(\phi)$. Since $\theta \mapsto \ell(\cdot\mid\theta)$ is constant on each $\phi$, there exists a $\mathcal{B}_\Phi$-measurable kernel $\hat \ell( y \mid \cdot)$ such that $\hat \ell(y  \mid \varphi(\theta)) = \ell (y \mid \theta)$ for all $y \in Y$ and $\theta \in \Theta$. To simplify notation, we let $\ell(\cdot \mid \phi) \equiv \hat \ell(\cdot \mid \varphi(\theta))$ whenever $\varphi(\theta) = \phi$. 

\begin{example}[Canonical Reduced-Form Space] \label{example:canonical}
    Every likelihood $\ell$ set-identifies $\theta$ up to its canonical reduced-form space $\Phi_\ell$, defined by $\varphi_\ell(\theta) := \ell(\cdot \mid \theta)$  and $\Phi_\ell := \varphi_\ell(\Theta) \subseteq \Delta(Y)$. Under this representation, each reduced-form parameter corresponds to a distinct sampling distribution, and the identified set $IS_\Theta(\phi) = \left\{ \theta \in \Theta : \ell(\cdot \mid \theta) = \phi \right\}$ corresponds to all observationally equivalent structural parameters.\footnote{Since $Y$ is Polish and $\theta \mapsto \ell(\cdot \mid \theta)$ is weakly continuous, $\varphi_\ell$ is continuous and all $IS_\Theta(\phi)$ are closed.}
\end{example}

It is well-known that Bayesian inference based on a single prior $\pi \in \Delta(\Theta)$ can be sensitive to the choice of prior over unidentified features of the structural parameter \citep[e.g.,][]{Moon_Schorfheide_ECTA_2012, Kline_Tamer_QE_2016}. Recent work by \cite{Giacomini_Kitagawa_ECTA_2021} addresses this issue by specifying a prior $\rho_0 \in \Delta(\Phi)$ over the reduced-form parameters instead and considering the set of structural priors obtained by combining $\rho_0$ with every possible conditional distribution over each identified set $\phi$. Equivalently, this prior set is the Aumann expectation of the correspondence $\phi \mapsto \Delta(\phi)$ under $\rho_0$, where $\Delta(\phi) \equiv\{\pi \in \Delta(\Theta): \pi(\phi) = 1\}$ is the sub-simplex over $\phi$.

Formally, let  $\bP_\Phi$ be a  random set defined on the probability space $(\Phi, \mathcal{B}_\Phi, \rho)$, with $\E_{\rho}[\bP_\Phi]$ denoting its Aumann expectation.\footnote{See, e.g., Chapter~3 of \cite{Molchanov_Molinari_2018} for the formal definition and properties of Aumann expectations. They coincide with $\rho$-weighted Minkowski sums for finite $\Phi$.}   The following notion of partially identified ambiguity  generalizes the full ambiguity approach of \cite{Giacomini_Kitagawa_ECTA_2021}. 

\begin{definition}[Partially Identified Set] \label{def:PI}
   A set $\Pi \in \mathcal{K}^c(\Delta(\Theta))$ is  partially identified by $(\rho, \Phi)$ if there exists a measurable correspondence $\bP_\Phi:\Phi\rightrightarrows\Delta(\Theta)$ such that
    \begin{align*}
        \Pi = \E_\rho[\bP_\Phi],
    \end{align*}
    where  $\bP_\Phi(\phi) =: \Pi_\phi \in \mathcal{K}^c(\Delta(\phi))$, $\rho$-almost surely.  If $\Pi_\phi = \Delta(\phi)$, $\rho$-almost surely, then  $\Pi$ is  the full ambiguity set  partially identified by $(\rho, \Phi)$.
\end{definition} 

Recall that in Section~\ref{sec:example}, we considered the full and restricted ambiguity sets $\Pi_0^F$ and $\Pi_0^R$ induced by the uniform prior $\rho_0$ over $\Phi=\{-2,2\}$. We show below that $\Pi_0^F$ is the full ambiguity set partially identified by $(\rho_0,\Phi)$, while $\Pi_0^R$ is not partially identified by $(\rho_0,\Phi)$. Section~\ref{sec:discussion} provides more general examples of prior sets commonly used in robust Bayesian analysis that can be formulated within the framework of partially identified ambiguity.

\begin{example} \label{example:motivating}
Consider the random set $\bP_\Phi$ defined by $\bP_\Phi(-2)=\Delta([-3,-1])$ and $\bP_\Phi(2)=\Delta([1,3])$. Since $\Phi$ is finite, the Aumann expectation of $\bP_\Phi$ reduces to the Minkowski sum
\begin{align*}
    \E_{\rho_0}[\bP_\Phi]
    &= \frac{1}{2}\Delta([-3,-1])+\frac{1}{2}\Delta([1,3])\\
    &= \left\{
    \frac{1}{2}\pi_-+\frac{1}{2}\pi_+:
    \pi_-\in\Delta([-3,-1]),\,
    \pi_+\in\Delta([1,3])
    \right\}
    =\Pi_0^F,
\end{align*}
as desired. As for the restricted ambiguity set, let $p_x$ denote the Dirac measure at $x \in \R$ and observe that $\frac{1}{2} p_{-3} + \frac{1}{2} p_3, \frac{1}{2} p_{-1} + \frac{1}{2} p_1 \in\Pi_0^R$. If $\Pi_0^R$ is the $\rho_0$-Aumann expectation of some random set, then the definition of Minkowski addition implies $\frac{1}{2}p_{-3} + \frac{1}{2}p_1 \in \Pi_0^R$. However, this  measure violates the  mean-zero restriction,  a contradiction.
\end{example}

Robust Bayesian analysis retains many of the desirable properties of Bayesian decision theory when the likelihood $\ell(\cdot \mid \theta)$ set-identifies $\theta$ up to (a coarsening of) $\Phi$ and the prior set $\Pi_0 \in \K(\Delta \Theta)$ is partially identified by $(\rho_0, \Phi)$. First, as was the case with the two prior sets in Section~\ref{sec:example}, every prior $\pi_0 \in \Pi_0$ induces the same prior predictive $\mu(\cdot) \equiv \int_\Theta \ell(\cdot \mid \theta) \pi_0(d\theta)$ and the following notion of consistency is  satisfied.

\begin{definition}[Consistency] \label{def_cons_ambig2}
    The likelihood  $\ell: \Theta \to \Delta(Y)$ is consistent with  $\Pi \in \K(\Delta \Theta)$ if for every $\pi, \pi' \in \Pi$,
    \begin{align}
   \int_\Theta  \ell( A \mid \theta) \,   \pi(d\theta) & =   \int_\Theta   \ell( A \mid \theta) \,  \pi'(d \theta), \qquad \forall A  \in \mathcal{Y}. \nonumber
    \end{align} 
\end{definition}

Moreover, every prior $\pi_0 \in \Pi_0$ induces the same marginal distribution $\rho_0$ over the reduced-form parameters.  These two results are collected in the following lemma.\footnote{Note that partial identification of the prior set $\Pi_0$ is not necessary for either property. For instance, the restricted ambiguity set $\Pi_0^R$ from Section~\ref{sec:example} satisfies both properties despite not being partially identified by $(\rho_0,\Phi)$. We provide a characterization of partially identified prior sets in Section~\ref{sec:main-results}.}

\begin{lemma} \label{lemma:consistent_PI}
    Suppose $\ell: \Theta \to \Delta(Y)$ set-identifies $\theta$ up to a coarsening of $\Phi$.  If $\Pi_0 \in \K(\Delta \Theta)$ is partially identified by $(\rho_0, \Phi)$, then $(\Pi_0, \ell)$ is consistent and $\pi_0 \circ \varphi^{-1} = \rho_0$ for all $\pi_0 \in \Pi_0$.  
\end{lemma}
\begin{proof}
    See Appendix~\ref{pf:consistent_PI}.
\end{proof}

The random set framework also provides a natural way to formulate our set-valued generalization of the rational expectations property. Let $\bP_Y: Y \rightrightarrows \Delta(\Theta)$ denote the posterior correspondence  $\bP_Y(y) = \Pi_y$, where $\Pi_y$ is obtained by updating $\Pi_0$ prior-by-prior under $\ell$. Under consistency, $\bP_Y$ defines a random set on the probability space $(Y,\mathcal{Y},\mu)$. The Aumann plausibility condition requires the Aumann expectation of $\bP_Y$ to equal the prior set.

\begin{definition}[Aumann Plausibility]
Suppose $(\Pi_0,\ell)$ is consistent, with common prior predictive distribution $\mu$. The pair $(\Pi_0,\ell)$ satisfies Aumann plausibility if
\begin{align*}
    \E_\mu[\bP_Y]=\Pi_0.
\end{align*}
\end{definition}

When $\Pi_0$ is a singleton, Aumann plausibility reduces to the rational expectations property (i.e., Bayes plausibility), which says that the  prior predictive expectation of posterior beliefs equals the prior under Bayesian updating \citep[e.g.,][]{KG_bayesian_persuasion}.  Section~\ref{sec:discussion} provides a preference-based motivation for Aumann plausibility and its normative appeal.

\begin{remark}[Consistency and Partial Identification]
Structural parameters within an identified set are observationally equivalent because they induce the same data likelihood. Consistency extends this observational equivalence to priors over $\Theta$, illustrating how a set-identified model gives rise to ambiguity between prior beliefs that the data cannot distinguish. To understand the informational content of consistency, suppose the  DM adopts a hierarchical Bayesian model with a hyperprior supported on the prior set $\Pi_0$. Since every $\pi_0 \in \Pi_0$ induces the same predictive distribution under consistency, Bayesian updating leaves the hyperprior unchanged after data realizations. In this sense, consistency captures the lack of quantifiable knowledge about the relative plausibility of priors in $\Pi_0$ that underlies the ambiguity in the multiple priors model.
\end{remark}

\section{Theoretical Results} \label{sec:main-results}

\subsection{Ex Ante Foundations for Conditional $\Gamma$-Minimax} \label{sec:ex-ante}

It is well known that the ex ante $\Gamma$-minimax rule obtained by minimizing \eqref{eq:gamma_MML} does not generally coincide with the conditional $\Gamma$-minimax rule obtained by minimizing \eqref{eq:c_gamma_MML} at each data realization $y \in Y$ \citep[e.g.,][]{vidakovic2000gamma, kitagawa2012gamma_minimax, Stoye_2012_AR, fernandez2026robust, Giacomini_Kitagawa_Read_2025}. However, this literature mostly focuses on counterexamples and numerical comparisons of the optimal rules under each criterion. In this section, we provide a direct comparison by constructing the joint belief set that provides an ex ante representation of the conditional $\Gamma$-minimax criterion.

Suppose the likelihood $\ell$ set-identifies $\theta$ up to (a coarsening of) $\Phi$. As with the two ambiguity sets in the introductory example, we assume that the likelihood is consistent with the prior set $\Pi_0$, such that there exists a common prior predictive distribution $\mu \in \Delta(Y)$. The ex ante representation of conditional $\Gamma$-minimax is then given by the set 
\begin{align}
    \Gamma^*(\Pi_0) = \left\{  \pi_y(d\theta) \mu(dy):  \pi_y \in \Pi_y  \text{ for }\mu\text{-a.e. }y    \right\}. \label{eq:gamma_rectangular_PI}
\end{align}
Formally, this is because $\Gamma^*(\Pi_0)$ is the rectangular hull of the set $\Gamma(\Pi_0)$, in the sense of \cite{ES_recursive}, with respect to the filtration generated by data realizations (see Appendix~\ref{Appendix:rectangularity} for details).\footnote{As the arrival of information is the only source of dynamic inconsistency, Corollary 3 of \cite{Hayashi_2011_TD} implies that the same result holds for the corresponding $\Gamma$-minimax regret criterion.}

To see why this representation rationalizes conditional $\Gamma$-minimax, observe that the string of equalities  \eqref{eq:Bayes-DC} in the Bayesian setting can be generalized to 
\begin{align}
    \inf_{\delta: Y \to \mathcal{A}} \sup_{\gamma \in \Gamma^*(\Pi_0)} \E_{\gamma}\left[ L(\theta,\delta(Y)) \right]  &=  \inf_{\delta: Y \to \mathcal{A}}   \E_{\mu}\left[ \sup_{\pi \in \Pi_Y} \E_{\pi}[L(\theta,\delta(Y))] \right]  \nonumber   
    \\ & =  \E_{\mu}\left[    \inf_{a \in \mathcal{A}} \sup_{\pi \in \Pi_Y} \E_{\pi}[L(\theta,a)] \right]. 
\end{align}
Since the inner minimization in the final expression is precisely the conditional $\Gamma$-minimax problem, any decision rule that is optimal under ex ante $\Gamma^*(\Pi_0)$-minimax prescribes a conditional $\Gamma$-minimax optimal action, $\mu$-almost surely.

We can now directly compare the ex ante and conditional $\Gamma$-minimax criteria through their corresponding sets $\Gamma(\Pi_0)$ and $\Gamma^*(\Pi_0)$ of joint distributions. The following proposition characterizes their relationship and identifies the restrictions that distinguish the two criteria.  This allows us to provide a belief-based interpretation of how the conditional criterion is more conservative than the ex ante criterion. For $\gamma \in \Delta(\Theta \times Y)$, note that  we write $Y \ind_\gamma \theta \mid \phi$ whenever $Y$ and $\theta$ are conditionally independent under $\gamma$ given $\phi$.\footnote{Formally, this means that $\gamma(Y \in A \mid \theta) = \gamma(Y \in A \mid \varphi(\theta))$ for every $A \in \mathcal{Y}$, $\gamma$-almost surely.}

\begin{proposition}
\label{prop:gamma_characterization}
Suppose the likelihood $\ell$ set-identifies $\theta$ up to a coarsening of $\Phi$ and every $\pi_0 \in \Pi_0$ induces the same reduced-form marginal $\rho_0 \in \Delta(\Phi)$, such that $(\Pi_0, \ell)$ is consistent. For any $\pi \in \Pi_0$, let $\pi(\cdot \mid \phi)$ denote the conditional distribution of $\theta$ given $IS_\Theta(\phi)$. Then,
\begin{align*}
    \Gamma(\Pi_0)
    & = \left\{ \gamma(d\theta,dy) = \int_\Phi \ell(dy\mid\phi)\pi_0(d\theta\mid\phi)\rho_0(d\phi) : \pi_0\in\Pi_0 \right\}, \\
    \Gamma^*(\Pi_0)
    & = \left\{ \gamma^*(d\theta,dy) = \int_\Phi \ell(dy\mid\phi)\pi_{0,y}(d\theta\mid\phi)\rho_0(d\phi) : \pi_{0,y}\in\Pi_0, \,\,  \mu\text{-a.e. } y \right\},
\end{align*}
where $\phi \mapsto \pi_0(\cdot\mid \phi)$ and  $(\phi,y) \mapsto \pi_{0,y}(\cdot\mid \phi)$ are measurable kernels. Consequently,
\begin{align*}
    \Gamma(\Pi_0) = \left\{ \gamma^*\in\Gamma^*(\Pi_0): Y\ind_{\gamma^*}\theta\mid\phi, \,\,\, \gamma^*(\cdot,Y)\in\Pi_0 \right\}.
\end{align*}
\end{proposition}
\begin{proof}
See Appendix~\ref{pf:gamma_characterization}.
\end{proof}

The robust Bayesian literature often describes conditional $\Gamma$-minimax as the more conservative or robust criterion \citep[e.g.,][]{christensen2026optimal}. Our result above makes this distinction explicit in terms of an adversarial Nature's set of strategies in the game underlying minimax analysis.  Under ex ante $\Gamma$-minimax, Nature must commit to a single conditional distribution $\pi_0(\cdot \mid \phi)$  over each identified set. On the other hand, conditional $\Gamma$-minimax allows Nature to select a profile $\{\pi_{0,y}(\cdot \mid \phi)\}_y$ of conditional distributions that can vary across data realizations. Note that joint beliefs under both criteria preserve the reduced-form likelihood $\ell(\cdot \mid \phi)$.

 \begin{remark}[Conservatism]
    Proposition~\ref{prop:gamma_characterization} identifies two distinct ways in which conditional $\Gamma$-minimax enlarges the set of joint beliefs. First, it relaxes the conditional independence restriction $Y \ind \theta \mid \phi$, allowing Nature's choice of the distribution of $\theta$ within each identified set to depend on the data realization. Second, it can implicitly introduce prior distributions over $\Theta$ that lie outside the DM's original prior set, such that $\gamma^*(\cdot , Y) \not \in \Pi_0$ as in the introductory example with $\Pi_0 = \Pi_0^R$.
    
    We interpret the first relaxation as genuine conservatism  relative to ex ante $\Gamma$-minimax because it gives Nature greater flexibility without changing the empirical content of the statistical model. In particular, the reduced-form likelihood $\ell(\cdot \mid \phi)$ is preserved, and the additional dependence has no testable implications because the distribution of $\theta$ within each identified set is unidentified.  By contrast, the second relaxation cannot be justified as conservatism because it introduces beliefs excluded by the DM's prior set  and leads to violations of the rational expectations property.\footnote{We provide a preference-based motivation for the Aumann plausibility condition in Section~\ref{sec:normative}.}
\end{remark}

\subsection{Main Equivalence Characterization}  \label{sec:main-equiv}

Recall that $\bP_Y: Y \rightrightarrows \Delta(\Theta)$ is the random set defined by $\bP_Y(y) = \Pi_y$, where $\Pi_y$ is the set of  posteriors  obtained by updating each prior in the prior set $\Pi_0$.   The set of marginal beliefs over $\Theta$ in $\Gamma^*(\Pi_0)$ coincides with the  $\mu$-Aumann expectation  of $\bP_Y$, since
\begin{align*}
    \Pi_0^*  &  := \left\{ \gamma^*(\cdot,Y): \gamma^*\in\Gamma^*(\Pi_0)  \right\} \\
    & = \left\{ \int_Y  \pi_y(\cdot) \, \mu(dy) : \pi_y \in \Pi_y, \,\,   \mu\text{-a.e.} \right\}   \equiv \E_\mu[\bP_Y]. 
\end{align*}  
Our discussion in the preceding section then implies that if $(\Pi_0, \ell)$ satisfies Aumann plausibility, then $\Pi_0^* = \Pi_0$ and conditional $\Gamma$-minimax admits a valid interpretation as a more conservative version of ex ante $\Gamma$-minimax. This section proves that Aumann plausibility holds if and only if $\Pi_0$ is partially identified by $(\rho_0,\Phi_\ell)$ for some $\rho_0 \in \Delta(\Phi_\ell)$.\footnote{Recall from Example~\ref{example:canonical} that $\Phi_\ell$ denotes the canonical reduced-form space corresponding to $\ell$.}

We first show that Bayesian updating preserves the structure of partially identified ambiguity. Specifically, if $\Pi_0$ is partially identified by $(\rho_0,\Phi_\ell)$, then the posterior set $\Pi_y$ is partially identified by $(\rho_y,\Phi_\ell)$, where $\rho_y$ is the reduced-form posterior obtained by updating $\rho_0$ under $\ell(\cdot\mid\phi)$. Updating therefore affects the prior set only through the reduced-form distribution, leaving the set of conditional distributions within each identified set unchanged.\footnote{For this reason, \cite{Giacomini_Kitagawa_ECTA_2021}  refer to the reduced-form prior $\rho_0$ as the \emph{revisable prior knowledge} and the set of conditional priors $\{\pi_0(\cdot | \phi) \}_{\phi \in \Phi}$ as the \emph{unrevisable prior knowledge} for full ambiguity sets. Lemma \ref{lemma:updating_equivalence} slightly extends this observation to all partially identified prior sets.}

 \begin{lemma} \label{lemma:updating_equivalence}
 Suppose $\ell$ set-identifies $\theta$ up to a coarsening of $\Phi$ and $\Pi_0$ is partially  identified by $(\rho_0, \Phi)$. Let  $\mu \in \Delta(Y)$ denote the corresponding prior predictive distribution. For $\mu$-almost every $y \in Y$, the posterior set  $\Pi_y$ is partially identified by $(\rho_y, \Phi)$, such that
    \begin{align*}
        \Pi_y = \E_{\rho_y}[\bP_{\Phi}]. 
    \end{align*}
 \end{lemma}
\begin{proof}
    See Appendix~\ref{pf:updating_equivalence}.    
\end{proof}

The preceding lemma provides the key insight behind the sufficiency of partially identified ambiguity for Aumann plausibility. Since Bayesian updating changes only the reduced-form distribution, the law of iterated (Aumann) expectations implies that averaging the posterior sets under the prior predictive $\mu$ recovers the original prior set $\Pi_0$. Our main characterization result establishes that the converse also holds. That is, partially identified ambiguity is not only sufficient, but also necessary for  Aumann plausibility.

\begin{theorem} \label{thm:two_ways}
Suppose $\Theta$ is compact. If $\Pi_0 \in \K^c(\Delta \Theta)$ and $\ell: \Theta \to \Delta(Y)$ are consistent, then the following are equivalent:
\begin{enumerate}
    \item There exists some $\rho_0 \in \Delta(\Phi_\ell)$ such that $\Pi_0$ is partially identified by $(\rho_0, \Phi_\ell)$.
    \item The pair $(\Pi_0, \ell)$ satisfies Aumann plausibility.
\end{enumerate}
\end{theorem}
\begin{proof}
    See Appendix~\ref{pf:two_ways}.
\end{proof}

The necessity of partially identified ambiguity for Aumann plausibility is particularly striking because it imposes substantial restrictions on the collection of permissible prior sets, ruling out prior sets like the restricted ambiguity set $\Pi_0^R$ from the motivating example. The proof proceeds by first showing that the collection of conditional distributions over identified sets is decomposable, then uses a Castaing representation argument to construct the random set $\bP_\Phi$ underlying the Aumann expectation representation of partially identified prior sets \citep[see, e.g., Chapter 3.1 of][]{Molchanov_Molinari_2018}. 

From a normative perspective, Theorem~\ref{thm:two_ways} establishes precisely when conditional $\Gamma$-minimax can be justified as the more conservative alternative to ex ante $\Gamma$-minimax. Under partially identified ambiguity, the joint belief set $\Gamma^*(\Pi_0)$ allows Nature greater flexibility while preserving the DM's original prior set $\Pi_0$. Conversely, without partially identified ambiguity, conditional $\Gamma$-minimax implicitly introduces prior beliefs outside $\Pi_0$. We therefore argue that conditional $\Gamma$-minimax should only be used under partially identified ambiguity. 


\section{Discussion} \label{sec:discussion}

\subsection{Behavioral Interpretation of Aumann Plausibility} \label{sec:normative}

We now provide a preference-based interpretation of Aumann plausibility. Under the multiple priors model of \cite{GS_MEU}, the DM's original prior set $\Pi_0$ represents her preferences over actions that depend on the structural parameter $\theta$. The ex ante representation of conditional $\Gamma$-minimax evaluates decision rules, which depend on both $\theta$ and $Y$, under the joint belief set $\Gamma^*(\Pi_0)$. A natural question is whether this extension preserves the DM's original preferences over outcomes that do not depend on the data realization.

To formalize this distinction, consider the decision problem induced by the action set $\mathcal{A}$ and loss function $L:\Theta \times \mathcal{A} \to \R$, assuming that $L(\cdot, a)$ is continuous and bounded at every $a \in \mathcal{A}$. That is, we restrict attention to constant decision rules (i.e., actions), such that losses only  depend on $\theta$. The worst-case expected losses under the original prior set $\Pi_0$ and the enlarged joint belief set $\Gamma^*(\Pi_0)$ are, respectively,
\begin{align*}
    \sup_{\gamma\in\Gamma(\Pi_0)}\E_\gamma[L(\theta,a)] & = \sup_{\pi_0\in\Pi_0}\E_{\pi_0}[L(\theta,a)], \\
    \sup_{\gamma^*\in\Gamma^*(\Pi_0)}\E_{\gamma^*}[L(\theta,a)] & = \sup_{\pi_0^*\in\Pi_0^*}\E_{\pi_0^*}[L(\theta,a)].
\end{align*}
As we show in the following lemma, this means that $\Pi_0$ and $\Gamma^*(\Pi_0)$  induce identical preferences over actions if and only if Aumann plausibility holds. 

\begin{lemma} \label{lemma:preference_preservation}
    Suppose $\Pi_0 \in \K^c(\Delta \Theta)$ and $(\Pi_0,\ell)$ is consistent. The following are equivalent:
    \begin{enumerate}
        \item $(\Pi_0,\ell)$ satisfies Aumann plausibility.
        \item For every finite action space $\mathcal{A}$ and continuous loss $L$, the $\Pi_0$-minimax and $\Gamma^*(\Pi_0)$-minimax criteria induce identical preferences over $\mathcal{A}$.
    \end{enumerate}
\end{lemma}
\begin{proof}
    See Appendix~\ref{pf:preference_preservation}.
\end{proof}
 
In other words, Aumann plausibility is equivalent to preserving the DM's original preferences over actions. This property is normatively desirable because conditional $\Gamma$-minimax can be viewed as the dynamically consistent extension of the DM's original preferences over actions to data-contingent decision rules (see Section~\ref{sec:ex-ante}). Under this interpretation, the extension should preserve the DM's preferences over actions, instead of implicitly revising the preferences represented by $\Pi_0$. Hence, Lemma~\ref{lemma:preference_preservation}  provides a behavioral justification for Aumann plausibility, complementing its interpretation as a set-valued generalization of the rational expectations property.

\subsection{Implications for Robust Bayesian Analysis}

We have argued thus far that conditional $\Gamma$-minimax should only be used under partially identified ambiguity. One way to satisfy this requirement is to adopt the full ambiguity specification of \cite{Giacomini_Kitagawa_ECTA_2021}. However, partially identified ambiguity also accommodates prior sets that impose additional restrictions on the conditional distributions within identified sets. This section provides two examples from the robust Bayes literature.

\paragraph{KL-Neighborhoods} We first consider the KL-neighborhood specification of \cite{giacomini2019estimation}. Fix a reduced-form prior $\rho_0 \in \Delta(\Phi)$, a benchmark conditional distribution $q_0(\cdot \mid \phi) \in \Delta(\phi)$ for each $\phi \in \Phi$, and a radius $\lambda \geq 0$. At each $\phi \in \Phi$, define the set of conditional distributions
\begin{align*}
    \Pi_\phi^{\lambda\text{-}KL} := \left\{  q \in \Delta(\phi): D_{KL}\big(q \,\|\, q_0(\cdot \mid \phi)\big) \leq \lambda \right\},
\end{align*}
where $D_{KL}$ denotes KL divergence. The corresponding prior set over $\Theta$ is
\begin{align*}
    \Pi_0^{\lambda\text{-}KL}  := \E_{\rho_0}[\bP_\Phi^{\lambda\text{-}KL}], \qquad  \bP_\Phi^{\lambda\text{-}KL}(\phi) = \Pi_\phi^{\lambda\text{-}KL}.
\end{align*}

Observe that   $\Pi_0^{\lambda\text{-}KL}$ consists of all structural priors that share the reduced-form marginal $\rho_0$ and have conditional distributions that lie within the KL neighborhood of the benchmark $q_0(\cdot \mid \phi)$ at each identified set. It follows by definition that $\Pi_0^{\lambda\text{-}KL}$ is partially identified by $(\rho_0, \Phi)$.  When $\lambda=0$, the prior set reduces to the benchmark prior. As $\lambda$ increases, the prior set expands, approaching the full ambiguity specification in the limit $\lambda \to \infty$ when the benchmark conditional distributions have full support on their respective identified sets.

As discussed in \cite{Giacomini_Kitagawa_Read_2025},  the above specification differs from constructing a KL neighborhood around an unconditional benchmark prior $\bar \pi_0 \in \Delta(\Theta)$ and taking
\begin{align*}
    \widetilde \Pi_0^{\lambda\text{-}KL} : = \left\{   \pi \in \Delta(\Theta):  D_{KL}(\pi \,\|\, \bar\pi_0) \leq \lambda  \right\}.
\end{align*}
Unlike $\Pi_0^{\lambda\text{-}KL}$, this prior set generally allows the reduced-form marginal to vary and therefore does not exhibit partially identified ambiguity. \cite{giacomini2019estimation} motivate their conditional KL-neighborhood specification as a way to isolate sensitivity to unrevisable prior beliefs. They also show that conditional $\Gamma$-minimax rules are analytically tractable under $\Pi_0^{\lambda\text{-}KL}$. Our characterization results  provide an additional normative justification for their approach, since Theorem~\ref{thm:two_ways} implies that $\Pi_0^{\lambda\text{-}KL}$ satisfies Aumann plausibility under any likelihood that set-identifies $\theta$ up to $\Phi$, while $\widetilde \Pi_0^{\lambda\text{-}KL}$ does not.

\paragraph{$\epsilon$-Contamination}  Another common robust Bayes specification is $\epsilon$-contamination \citep[e.g.,][]{hodges1952use}. Given a benchmark prior $\bar\pi_0 \in \Delta(\Theta)$, set of contamination distributions $Q \in \K^c(\Delta \Theta)$, and $\epsilon \in [0,1]$, the corresponding prior set takes the form
\begin{align*}
    \Pi_0^\epsilon = (1-\epsilon)\{\bar\pi_0\} + \epsilon Q,
\end{align*}
where $\epsilon$ controls the degree of ambiguity around the benchmark prior.

In the context of partial identification, \cite{giacomini2022uncertain} propose a model-averaging approach that can be interpreted as $\epsilon$-contamination, where the benchmark prior comes from a single-prior model and the contamination distributions correspond to a set-identified multiple priors model. To illustrate their approach, suppose the two models share the same reduced-form prior $\rho_0 \in \Delta(\Phi)$, and let $q_0(\cdot \mid \phi) \in \Delta(\phi)$ denote the benchmark conditional distribution within each identified set. Let  $Q_\phi= \Delta(\phi)$ be the set of contamination distributions at each $\phi \in \Phi$.   Given the set of conditional distributions  
\begin{align*}
    \Pi_\phi^\epsilon = \left\{ (1-\epsilon)q_0(\cdot \mid \phi) + \epsilon q : q \in Q_\phi \right\},
\end{align*}
we can then express the prior set as
\begin{align*}
    \Pi_0^\epsilon = (1-\epsilon) \{\bar\pi_0\} + \epsilon \E_{\rho_0}[\bP^1_\Phi] = \E_{\rho_0} [\bP_\Phi^\epsilon], \qquad \bP_\Phi^\epsilon(\phi) := \Pi_\phi^\epsilon,  
\end{align*}
where  $\bar\pi_0 = \int_\Phi q_0(\cdot \mid \phi)\rho_0(d\phi)$. Note that $\epsilon=0$ corresponds to the Bayesian setting, while $\epsilon=1$ reduces to the full ambiguity specification. Since $\Pi_0^\epsilon$ is partially identified by $(\rho_0, \Phi)$, our characterization result provides a normative justification for its use in conditional $\Gamma$-minimax decision making.

\subsection{Value of Information} \label{sec:blackwell}

The preceding examples illustrate how common robust Bayesian specifications can be formulated under partially identified ambiguity. We now show that conditional $\Gamma$-minimax decision making under partially identified ambiguity also preserves a fundamental property of Bayesian decision making: more informative statistical experiments are more valuable.

Following \cite{Blackwell_original, Blackwell_original_2}, we now refer to the likelihoods $\ell: \Theta \to \Delta(Y)$ as experiments. Given a prior set $\Pi_0 \in \K^c(\Delta\Theta)$ and an experiment $\ell$ consistent with $\Pi_0$, let $\Pi_y^\ell$ denote the posterior set obtained by prior-by-prior updating and $\mu^\ell \in \Delta(Y)$ the common prior predictive distribution. Given an action set $\mathcal{A}$ and loss $L:\Theta \times \mathcal{A} \to \R$, we define the value of an experiment as the expected utility (i.e., negative loss) from choosing the optimal conditional $\Gamma$-minimax action at each data realization, allowing for either deterministic or randomized actions.\footnote{For non-singleton $\Pi_0$,  randomization can improve the DM's worst-case payoff, so we have $W_{\Pi_0,\mathcal A,L}(\ell)\geq V_{\Pi_0,\mathcal A,L}(\ell)$, potentially with strict inequality.} With singleton $\Pi_0$, both value functions coincide with the usual Bayesian value of information. 

\begin{definition}[Value of Information]
    Given $(\Pi_0,\mathcal A,L)$ and an experiment $\ell$ consistent with $\Pi_0$, the value of $\ell$ under deterministic and randomized actions, respectively, are 
    \begin{align}
        V_{\Pi_0, \mathcal{A}, L}(\ell) & :=  -\E_{\mu^\ell} \bigg[ \inf_{a \in \mathcal{A}} \sup_{\pi_y \in \Pi_y^\ell} \E_{\pi_y}[L(\theta,a)] \bigg] \nonumber\\
        W_{\Pi_0, \mathcal{A},L}(\ell) & :=  -\E_{\mu^\ell} \bigg[ \inf_{\sigma \in \Delta(\mathcal{A})} \sup_{\pi_y \in \Pi_y^\ell} \E_{\sigma} \left[ \E_{\pi_y}[L(\theta,a)] \right] \bigg].  \nonumber
    \end{align}
\end{definition}

 Next, we define the more informative order over the collection of experiments. The definition is identical to that in the Bayesian setting, capturing the intuition that any information in $\ell_2$ can be generated by adding noise to $\ell_1$ whenever $\ell_1$ is  more informative than $\ell_2$.

\begin{definition}[More Informative]
    An experiment $\ell_1: \Theta \to \Delta(Y_1)$ is more informative than $\ell_2: \Theta \to \Delta(Y_2)$ if  there exists a kernel $k:Y_1 \to \Delta(Y_2)$ such that
    \begin{align*}
        \ell_2 (B \mid \theta)=\int_{Y_1} k(B \mid  y_1) \ell_1( dy_1 \mid \theta), \qquad \forall B \in \mathcal Y_2, \,\, \theta\in\Theta.
    \end{align*}
\end{definition}

When the prior set $\Pi_0$ is a singleton, Blackwell's theorem of informativeness says that the more informative and value function orders are equivalent.  The following proposition extends this result to the setting of partially identified ambiguity, where the key step is to reformulate the problem in the Bayesian setting using the reduced-form parameters.\footnote{ As a result, our results admit relatively simple proofs and are different in spirit to much of the existing work that extends  Blackwell's theorem to decision making under ambiguity \citep[e.g.,][]{celen_blackwell, renou_blackwell, li_zhou_blackwell, wang_blackwell}.}

\begin{proposition} 
    \label{prop:Blackwell}
   For any $\ell_1:\Theta\to\Delta(Y_1)$ and $\ell_2:\Theta\to\Delta(Y_2)$ that set-identify $\theta$ up to a coarsening of $\Phi$, the following are equivalent:
    \begin{enumerate}
        \item $\ell_1$ is more informative than $\ell_2$.
        \item $V_{\Pi_0,\mathcal A,L}(\ell_1)\geq V_{\Pi_0,\mathcal A,L}(\ell_2)$ for every $(\mathcal{A}, L)$ and  $\Pi_0$ partially identified by some $(\rho_0,\Phi)$.
        \item $W_{\Pi_0,\mathcal A,L}(\ell_1)\geq W_{\Pi_0,\mathcal A,L}(\ell_2)$ for every $(\mathcal{A}, L)$ and  $\Pi_0$ partially identified by some $(\rho_0,\Phi)$.
    \end{enumerate}
    Moreover, the inequalities hold with equality for singleton $\mathcal{A}$.
\end{proposition}
\begin{proof}
See Appendix~\ref{pf:Blackwell}.
\end{proof}

The above result establishes that Blackwell's ordering of statistical experiments is preserved under conditional $\Gamma$-minimax with partially identified ambiguity, regardless of whether the DM can randomize over actions. This includes the KL-neighborhood and $\epsilon$-contamination specifications discussed above. More broadly, Proposition~\ref{prop:Blackwell} complements our earlier characterization of Aumann plausibility by illustrating how conditional $\Gamma$-minimax retains desirable properties of Bayesian decision theory under partially identified ambiguity.
 
\section{Conclusion} \label{sec:conclusion}

This paper studies the distinction between ex ante and conditional $\Gamma$-minimax criteria in robust Bayesian analysis. By providing an ex ante representation of conditional $\Gamma$-minimax through the joint belief set $\Gamma^*(\Pi_0)$, we characterize the additional beliefs implicit in conditional decision making. We then introduce Aumann plausibility as the set-valued generalization of the rational expectations property  and show that it holds if and only if the DM's prior set exhibits partially identified ambiguity. We interpret this result as suggesting that conditional $\Gamma$-minimax should only be used under partially identified ambiguity. We further show that common robust Bayesian specifications, including KL-neighborhoods and $\epsilon$-contamination, can be accommodated within this framework, and that conditional $\Gamma$-minimax preserves Blackwell's ordering of statistical experiments under partially identified ambiguity. Together, these results provide normative guidance for the use of conditional $\Gamma$-minimax in robust Bayesian decision making under partial identification.

\appendix
 
\numberwithin{definition}{section}
\numberwithin{lemma}{section}
\numberwithin{proposition}{section}

\section{Rectangularity} \label{Appendix:rectangularity}

The distinction between the ex ante and conditional $\Gamma$-minimax criteria is closely related to the notion of \emph{rectangularity} introduced by \cite{ES_recursive}.  Consider the two-period filtration $\{\mathcal{F}_t\}_{t=0}^2$ on $\Theta \times Y$, where $\mathcal{F}_0$ is trivial, $\mathcal{F}_1 = \sigma(Y)$, and $\mathcal{F}_2 = \mathcal{B}_\Theta \otimes \mathcal{Y}$. This filtration represents the information structure of our decision problem, where the DM first observes the data and subsequently learns the true parameter. 

Given a set of joint beliefs $\mathcal{Q} \subseteq \Delta(\Theta \times Y)$, let $M(\mathcal{Q})$ denote its set of marginal distributions over $Y$, and let $\mathcal{Q}_y$ denote the set of conditional distributions over $\Theta$ given the realization of $y$.

\begin{definition}[Rectangularity]
A set of joint beliefs $\mathcal{Q}\subseteq\Delta(\Theta\times Y)$ is rectangular with respect to $\{\mathcal{F}_t\}_{t=0}^2$ if
\begin{align*}
    \mathcal{Q}  =   \left\{  q_y(d\theta) \mu(dy) :  \mu \in M(\mathcal{Q}),\,  q_y\in\mathcal{Q}_y \text{ for } \mu\text{-almost every }y  \right\},
\end{align*}
where $y\mapsto q_y$ is a measurable probability kernel.
\end{definition}

Rectangularity therefore requires that any marginal distribution over $Y$ can be combined with any measurable selection of conditional beliefs from the corresponding posterior sets.  Now consider the joint belief set $\Gamma(\Pi_0)$  defined in \eqref{eq:gamma_joint} that  corresponds to the ex ante $\Gamma$-minimax criterion.  Let $M$ denote the set of prior predictive distributions induced by $\Pi_0$ and the likelihood $\ell$, such that
\begin{align*}
    M  \equiv \left\{ \mu(dy)  = \int_\Theta \ell(dy\mid\theta)\pi_0(d\theta) :  \pi_0 \in \Pi_0 \right\}.
\end{align*}
Note that $M$ is a singleton whenever $(\Pi_0, \ell)$ is consistent. The smallest set rectangular with respect to $\{\F_t\}$ that contains $\Gamma(\Pi_0)$ is its \emph{rectangular hull}, which is given by 
\begin{align}
    \Gamma^*(\Pi_0) := \left\{ \pi_y(d\theta) \mu(dy):  \pi_y \in \Pi_y, \, \mu\in M  \, \text{ for } \, \mu\text{-a.e. }y  \right\} \supseteq \Gamma(\Pi_0). \label{eq:gamma_rectangular}
\end{align}
As in the motivating example of Section~\ref{sec:example}, the key distinction is that $\Gamma(\Pi_0)$ requires posterior beliefs across data realizations to arise from a common prior $\pi_0 \in \Pi_0$, while $\Gamma^*(\Pi_0)$ allows them to arise from different priors.\footnote{This means that for $\gamma^* \in \Gamma^*(\Pi_0)$, there can exist distinct $y,y' \in Y$ and  $\pi_0, \pi_0' \in \Pi_0$ such that $\gamma^*(d\theta, dy) = \ell(dy\mid \theta)\pi_0(d\theta)$ and $\gamma^*(d\theta, dy') = \ell(dy' \mid \theta)\pi_0'(d\theta)$.}

As shown by \cite{ES_recursive}, rectangularity restores dynamic consistency between ex ante and conditional decision making. In particular, we can replicate the argument  \eqref{eq:Bayes-DC} from the Bayesian case, since
\begin{align*}
    \inf_{\delta: Y \to \mathcal{A}} \sup_{\gamma \in \Gamma^*(\Pi_0)} \E_{\gamma}\left[ L(\theta,\delta(Y)) \right]  &=  \inf_{\delta: Y \to \mathcal{A}}   \sup_{\mu \in M}  \E_{\mu}\left[ \sup_{\pi \in \Pi_Y} \E_{\pi}[L(\theta,\delta(Y))] \right] \nonumber\\
    & = \sup_{\mu \in M} \E_{\mu}\left[    \inf_{a \in \mathcal{A}} \sup_{\pi \in \Pi_Y} \E_{\pi}[L(\theta,a)] \right], 
\end{align*}
and the conditional $\Gamma$-minimax rule minimizes the integrand of the final expression at every $y \in Y$. This means that any conditional $\Gamma$-minimax rule is ex ante optimal under $\Gamma^*(\Pi_0)$-minimax. We can therefore interpret the joint belief set $\Gamma^*(\Pi_0)$ as providing the ex ante foundations of the conditional $\Gamma$-minimax criterion.

We remark that the robust Bayes literature often refers to the joint prior set $\Gamma(\Pi_0)$ as rectangular without specifying the underlying filtration. This terminology can be misleading, since rectangularity is defined relative to a filtration.   In particular, $\Gamma(\Pi_0)$ is rectangular with respect to the filtration $\{\mathcal{G}_t\}_{t=0}^2$, where $\mathcal{G}_0$ is trivial, $\mathcal{G}_1=\sigma(\theta)$, and $\mathcal{G}_2=\mathcal{B}_\Theta\otimes\mathcal{Y}$.   However, rectangularity with respect to this filtration does not imply rectangularity with respect to $\{\mathcal{F}_t\}_{t=0}^2$, which is the relevant information structure with data realizations. In fact, it follows from   \cite{Ellis_GEB_2018} that under a full support likelihood, no joint prior set can be rectangular with respect to both $\{\F_t\}_{t=0}^2$ and $\{\mathcal{G}_t\}_{t=0}^2$ unless   $\Pi_0$ is a singleton.

\section{Omitted Proofs} \label{Appendix:proofs}

\setcounter{lemma}{0}%
\renewcommand{\thelemma}{A.\arabic{lemma}}

\subsection{Proof of Lemma \ref{lemma:consistent_PI}} \label{pf:consistent_PI}

Suppose $\Pi_0$ is partially identified by $(\rho_0, \Phi)$, such that $\Pi_0 = \E_{\rho_0}[\bP_\Phi]$.  We first show that $\pi_0 \circ \varphi^{-1} = \rho_0$ for any $\pi_0 \in \Pi_0$.  Fix some $\pi_0 \in \Pi_0$. By definition of Aumann expectations, there exists a measurable selection $\phi \mapsto \pi_\phi$ satisfying
\begin{align*}
    \pi_0(A)  = \int_\Phi \pi_\phi(A) \, \rho_0(d\phi), \qquad \forall A \in \mathcal{B}_\Theta.
\end{align*}
Since $\Pi_\phi \in \K^c(\Delta( \phi))$ $\rho_0$-almost surely, $\pi_\phi(\phi) \equiv \pi_\phi \circ \varphi^{-1}(\phi) = 1$. For any $B \in \mathcal{B}_\Phi$, $\phi \in B$ implies $IS_\Theta(\phi) \subseteq \varphi^{-1}(B)$, while $\phi \not \in B$ implies $IS_\Theta(\phi) \cap \varphi^{-1}(B) = \emptyset$. That is, $\pi_\phi \circ \varphi^{-1}(B) = \1_B(\phi)$ for $\rho_0$-almost every $\phi$. Plugging into the above expression for $\pi_0$, we obtain
\begin{align*}
    \pi_0 \circ \varphi^{-1}(B) = \int_\Phi \pi_\phi\left( \varphi^{-1}(B)  \right) \rho_0(d\phi) = \int_\Phi \1_B(\phi) \, \rho_0(d\phi) = \rho_0(B).
\end{align*}

Next, we show that the pair $(\Pi_0, \ell)$ is consistent. For any $A \in \mathcal{Y}$ and $\pi_0 \in \Pi_0$, 
\begin{align*}
        \int_\Theta \ell (A \mid \theta) \, \pi_0 (d\theta) & = \int_\Theta \ell (A \mid \varphi(\theta)) \, \pi_0(d\theta)    \\
        & = \int_\Phi \ell (A \mid  \phi ) \, (\pi_0 \circ \varphi^{-1})(d\phi)  \\
        & = \int_\Phi \ell(A \mid \phi) \, \rho_0(d\phi),  
\end{align*}
where the last equality follows from the previous result. Since the final term does not depend on the choice of $\pi_0 \in \Pi_0$,  we have that $\ell$ and $\Pi_0$ are consistent, as desired.     \hfill $\blacksquare$

\subsection{Proof of Proposition~\ref{prop:gamma_characterization}} \label{pf:gamma_characterization}

As for the first equality, recall from \eqref{eq:gamma_joint} that
\begin{align*}
     \Gamma(\Pi_0)  &  \equiv \left\{  \gamma(d\theta, dy) = \pi_0(d\theta) \ell(dy \mid \theta) :  \pi_0 \in \Pi_0 \right\} \\ 
     & = \left\{ \gamma(d\theta, dy) = \ell(dy \mid \varphi(\theta)) \int_\Phi \pi_0(d\theta \mid \phi) \rho_0(d\phi) : \pi_0 \in \Pi_0  \right\} \\
     & = \left\{ \gamma(d\theta,dy) = \int_\Phi \ell(dy\mid\phi)\pi_0(d\theta\mid\phi)\rho_0(d\phi) : \pi_0\in\Pi_0 \right\},
\end{align*}
where the last step holds by noting that $\pi_0(d\theta \mid \phi) \equiv \pi_0(d\theta \mid  IS_\Theta(\phi)) = 0$ whenever $\varphi(\theta) \neq \phi$. 

For the second equality, let $\rho_y$ denote the reduced-form posterior given $y \in Y$. Since every $\pi_0 \in \Pi_0$ induces the same $\rho_0$, Bayes rule implies
\begin{align}
    \Pi_y = \left\{ \int_\Phi \pi_0(d\theta\mid\phi)\rho_y(d\phi) : \pi_0\in\Pi_0 \right\} \label{eq:cond-rep}
\end{align}
for $\mu$-a.e. $y$. Then, it follows from \eqref{eq:gamma_rectangular_PI} that 
\begin{align*}
    \Gamma^*(\Pi_0)
    & = \left\{ \pi_y(d\theta)\mu(dy) : \pi_y\in\Pi_y, \,\, \mu\text{-a.e. }y \right\} \\
    & = \left\{ \int_\Phi \pi_{0,y}(d\theta\mid\phi) \mu(dy) \rho_y(d\phi) : \pi_{0,y}\in\Pi_0, \,\,  \mu\text{-a.e. }y \right\} \\
    & = \left\{ \int_\Phi \ell(dy\mid\phi)\pi_{0,y}(d\theta\mid\phi)\rho_0(d\phi) : \pi_{0,y}\in\Pi_0, \,\, \mu\text{-a.e. }y \right\},
\end{align*}
where the last step follows from the fact that  $\mu(dy)  \rho_y(d\phi) = \ell(dy \mid \phi) \rho_0(d\phi)$.



It remains to prove the final characterization of $\Gamma(\Pi_0)$.  Since $\Gamma^*(\Pi_0)$ is the rectangular hull of $\Gamma(\Pi_0)$, we immediately have $\Gamma(\Pi_0) \subseteq \left\{    \gamma^*\in\Gamma^*(\Pi_0):   Y \ind_{\gamma^*} \theta \mid \phi, \,\,  \gamma^*(\cdot,Y) \in \Pi_0  \right\}$.   As for the reverse inclusion, consider any $\gamma^* \in \Gamma^*(\Pi_0)$ satisfying the two restrictions, and let $\pi_0 \equiv \gamma^*(\cdot,Y) \in \Pi_0$.   By \eqref{eq:cond-rep}, every posterior in $\Pi_y$ induces the same reduced-form posterior $\rho_y$, and  every $\gamma^*\in\Gamma^*(\Pi_0)$ satisfies
\begin{align*}
    \gamma^*( \varphi^{-1}(d \phi) ,dy)  =  \rho_y(d\phi)\mu(dy)  =   \rho_0(d\phi)\ell(dy\mid\phi).
\end{align*}
This means that the joint priors in $\Gamma^*(\Pi_0)$ preserve the structural likelihood $\ell(\cdot \mid \phi)$, with $\gamma^*(dy \mid \phi) = \ell(dy\mid\phi)$.   Together with $Y \ind_{\gamma^*} \theta \mid \phi$, this implies 
\begin{align*}
    \gamma^*(dy\mid\theta)  = \gamma^*(dy \mid \theta, \varphi(\theta))  = \gamma^*(dy \mid \varphi(\theta) )  = \ell(dy \mid \varphi(\theta))  = \ell (dy \mid \theta), 
\end{align*}
in which case $\gamma^*(d\theta,dy) = \pi_0(d\theta) \ell(dy\mid\theta)  \in \Gamma(\Pi_0)$, as desired. \hfill $\blacksquare$

\subsection{Proof of Lemma~\ref{lemma:updating_equivalence}} \label{pf:updating_equivalence}

 Fix $y \in Y$. For any $\pi_0 \in \Pi_0$,  there exists a measurable selection $\phi \mapsto \pi_\phi$ with  $\pi_\phi(\phi) = 1$  satisfying $\pi_0(A) = \int_\Phi \pi_\phi(A) \rho_0(d\phi)$ for all $A \in \mathcal{B}_\Theta$. Since $\rho_y \ll \rho_0$, any measurable selection of $\bP_\Phi$ under $\rho_0$ remains a valid measurable selection under $\rho_y$.

Define the predictive density $m(y) \equiv \int_\Phi \ell(y\mid \phi) \rho_0(d\phi)$, such that $\mu(dy) = m(y) dy$. For $\mu$-almost every $y$, the posterior update of $\pi_0$ is given by
\begin{align*}
    \pi_y(A)  \equiv \frac{1}{m(y)}  \int_A \ell(y \mid \theta) \pi_0(d\theta) & = \frac{1}{m(y)} \int_\Phi \int_{A \cap IS_\Theta(\phi)} \ell(y \mid \theta)  \, \pi_\phi(d\theta) \,  \rho_0( d\phi)   \\
    & = \frac{1}{m(y)} \int_\Phi \ell(y \mid \phi) \pi_\phi(A) \,  \rho_0(d\phi) \\
    & =   \int_\Phi  \pi_\phi(A) \, \rho_y(d\phi)
\end{align*}
for all $A \in \mathcal{B}_\Theta$. This implies that $\pi_y \in \E_{\rho_y}[\bP_\Phi]$, and thus $\Pi_y \subseteq \E_{\rho_y}[\bP_\Phi]$.
 
Conversely, consider arbitrary $\eta = \int_\Phi \pi_\phi \, \rho_y(d\phi) \in \E_{\rho_y}[\bP_\Phi]$  such that $\pi_\phi \in \Pi_\phi \equiv \bP_\Phi(\phi)$, $\rho_y$-almost surely. Fix some measurable selection $\phi \mapsto \pi^0_\phi$  and define
\begin{align*}
    \tilde \pi_\phi = \begin{cases}
        \pi_\phi & \text{if } \ell(y \mid \phi) > 0 \\
        \pi^0_\phi & \text{else.}
    \end{cases}
\end{align*}
Note that over $S_y \equiv \{\phi \in \Phi: \ell(y \mid \phi) > 0 \}$, the measures $\rho_y$ and $\rho_0$ have the same null sets. This means that $\tilde \pi_\phi \in \Pi_\phi$, $\rho_0$-almost surely, in which case $\tilde \pi_0 \equiv \int_\Phi \tilde \pi_\phi \, \rho_0(d\phi) \in \Pi_0$. From above, we know that the Bayesian update of $\tilde \pi_0$ satisfies
\begin{align*}
    \tilde \pi_y(A) = \int_\Phi \tilde \pi_\phi(A) \, \rho_y(d\phi)  & = \int_{S_y} \tilde \pi_\phi(A) \, \rho_y(d\phi)  \\
    & =  \int_{S_y}  \pi_\phi(A) \, \rho_y(d\phi)  = \int_{\Phi}  \pi_\phi(A) \, \rho_y(d\phi)  = \eta(A)
\end{align*}
for all $A \in \mathcal{B}_\Theta$. It follows that $\eta =  \tilde \pi_y \in \Pi_y$, allowing us to conclude that $\Pi_y \supseteq \E_{\rho_y}[\bP_\Phi]$, and thus $\Pi_y = \E_{\rho_y}[\bP_\Phi]$, as desired. \hfill $\blacksquare$

\subsection{Proof of Theorem~\ref{thm:two_ways}} \label{pf:two_ways}
 
($1.\Rightarrow 2.$)  Let $\Phi \equiv \Phi_\ell$.  Denote the support function corresponding to any $\Pi \in \K^c(\Delta \Theta)$ by 
\begin{align*}
    h_\Pi(f) \equiv \sup_{\pi \in \Pi} \int_\Theta f(\theta) \, \pi(d\theta), \qquad \forall f \in C_b(\Theta),
\end{align*}
where $C_b(\Theta)$ denotes the collection of continuous and bounded $f: \Theta \to \R$.

Since $\Pi_0$ is partially identified by $(\rho_0, \Phi)$, there exists a measurable correspondence $\bP_\Phi: \Phi \rightrightarrows \Delta(\Theta)$ such that $\Pi_0 = \E_{\rho_0}[\bP_\Phi]$. Because support functions commute with Aumann expectations for random sets that are  almost surely convex \citep[Theorem 3.11,][]{Molchanov_Molinari_2018}, it follows that
\begin{align*}
   \sup_{\pi_0 \in \Pi_0}  \int_\Theta f(\theta) \, \pi_0(d\theta)  \equiv  h_{\Pi_0}(f)  & = h_{ \E_{\rho_0}[\bP_\Phi] }(f) = \E_{\rho_0} \left[ h_{\Pi_\phi}(f) \right].
\end{align*}
By Lemma~\ref{lemma:updating_equivalence}, we know $\Pi_y = \E_{\rho_y}[\bP_\Phi]$ for $\mu$-almost every $y \in Y$. The same argument implies
\begin{align*}
    \sup_{\pi_y \in \Pi_y}  \int_\Theta f(\theta) \, \pi_y(d\theta)  \equiv h_{ \Pi_y}(f) = h_{\E_{\rho_y}[\bP_\Phi] }(f) = \E_{\rho_y} \left[ h_{\Pi_\phi}(f)  \right]. 
\end{align*}

Consider the joint probability space $(\Phi \times Y, \mathcal{B}_\Phi \times \mathcal{Y}, Q)$, where $Q(d\phi, dy) \equiv \rho_0(\phi) \ell(dy \mid \phi) = \rho_y(d\phi) \mu(dy)$ by Bayesian updating. Over this joint probability space, we have
\begin{align*}
      \sup_{\pi_0 \in \Pi_0}  \int_\Theta f(\theta)  \, \pi_0(d\theta)   &  = \E_{\rho_0} \left[ h_{\Pi_\phi}(f) \right] = \E_{Q} \left[ h_{\Pi_\phi}(f) \right] \\ 
      \sup_{\pi_y \in \Pi_Y}  \int_\Theta f(\theta) \pi_y(d\theta)  &  = \E_{\rho_y} \left[ h_{\Pi_\phi}(f)  \right] = \E_Q\left[ h_{\Pi_\phi}(f) \mid Y \right].
\end{align*}

Now consider any $\eta = \int_Y \pi_y \, \mu(dy) \in \E_\mu[\bP_Y]$, such that $y \mapsto \pi_y$ is a measurable selection of the posterior correspondence. We see that for any $f \in C_b(\Theta)$, 
\begin{align*}
    \int_\Theta f(\theta) \, \eta(d\theta) & = \int_Y  \left( \int_\theta f(\theta)  \, \pi_y(d\theta) \right) \mu(dy) \\
    & \leq \int_Y \left( \sup_{\pi_y \in \Pi_y}  \int_\theta f(\theta)  \, \pi_y(d\theta)  \right)    \mu(dy) \\
    & = \E_Q\left[ \E_Q\left[ h_{\Pi_\phi}(f) \mid Y \right] \right] \\
    & =  \E_Q\left[ h_{\Pi_\phi}(f)  \right] \\
    & = \sup_{\pi_0 \in \Pi_0} \int_\Theta  f(\theta) \, \pi_0(d\theta).
\end{align*}
The Hahn–Banach separation theorem then implies  $\eta \in \Pi_0$, and thus $\E_\mu[\bP_Y] \subseteq \Pi_0$.

Now fix any $\pi_0  \in \Pi_0$. The posterior updates 
$\pi_y$ of $\pi_0$ are in $\Pi_Y$, $\mu$-almost surely, so there exists a measurable selection $y \mapsto \pi_y$ of the posterior correspondence such that $\pi_0 = \int_Y \pi_y \, d\mu$. This implies  $\E_\mu[\bP_Y] \supseteq \Pi_0$, so we have that $\E_\mu[\bP_Y] = \Pi_0$, as desired.

(2. $\Rightarrow$ 1.) Let $\varphi(\theta) \equiv \ell(\cdot \mid \theta) \in \Delta(Y)$, such that the canonical reduced-form parameter space is given by  $\Phi \equiv \Phi_\ell = \varphi(\Theta)$.   By weak continuity of $\ell$, $\Phi$ is a separable metric, $\varphi$ is a continuous surjection, and the identified sets $IS_\Theta(\phi) \equiv \varphi^{-1}(\phi)$ are closed.

For any measurable $w: \Phi \to [0,1]$ and $\pi, \pi' \in \Delta(\Theta)$, define the mixture
\begin{align*}
    \pi \oplus_w \pi' \equiv (w \circ \varphi) \pi + \left( (1-w) \circ \varphi \right) \pi',
\end{align*}
such that $(\pi \oplus_w \pi')(d\theta) = w(\varphi(\theta)) \pi(d\theta) + (1-w(\varphi(\theta)) \pi'(d\theta)$. We first show that the prior set $\Pi_0$ must be closed under this mixture operation with $\theta$-dependent weights.

\begin{lemma} \label{lemma:likelihood}
    For any $A \in \mathcal{Y}$, define $w_A(\phi) \equiv \ell(A \mid \phi)$. If $\pi_0, \pi_0' \in \Pi_0$, then $\pi_0 \oplus_{w_A} \pi'_0 \in \Pi_0$. 
\end{lemma}
\begin{proof}
    Let $\pi_y$ and $\pi_y'$ denote the posterior updates of $\pi_0, \pi_0'$ given $y$.\footnote{Formally, $y \mapsto \pi_y$, $y \mapsto \pi_y'$ are the measurable posterior kernels of $\pi_0, \pi_0'$.} Since $\1_A(y) \pi_y + \1_{A^c}(y)  \pi_y'$ is a measurable selection of the random set $\bP_Y$ induced by the posterior sets, 
    \begin{align*}
        \pi_0'' \equiv \int_{Y} \left( \1_A(y) \pi_y + \1_{A^c}(y) \pi_y' \right) \, \mu(dy)  \in \E_\mu[\bP_Y] = \Pi_0,
    \end{align*}
    where the last equality holds by assumption of Aumann plausibility. For any $D \in \mathcal{B}_\Theta$,
    \begin{align*}
         \pi_0''(D) & = \int_A \pi_y(D) \, \mu(dy) + \int_{A^c} \pi_y'(D) \, \mu(dy) \\
         & = \int_D \ell(A \mid \theta) \, \pi_0(d\theta) + \int_D \ell(A^c \mid \theta) \, \pi_0'(d\theta) \\
         & = \int_D ( w_A \circ \varphi)(\theta) \, \pi_0(d\theta) + \int_D ( (1-w_A) \circ \varphi)(\theta) \, \pi_0'(d\theta) \\
         & = (\pi_0 \oplus_{w_A} \pi'_0 ) (D).
    \end{align*}
    Since our choice of $D$ was arbitrary, we have that $\pi_0 \oplus_{w_A} \pi'_0 = \pi_0'' \in \Pi_0$, as desired.
\end{proof}

The preceding lemma tells us that the prior set $\Pi_0$ is closed under mixtures using likelihood-induced weights. We define the general collection of measurable weights that preserve membership in $\Pi_0$ as
\begin{align*}
    \mathcal{W} \equiv \left\{ w : \Phi \to [0,1] : \pi_0 \oplus_w \pi'_0 \in \Pi_0, \,\,\, \forall \pi_0, \pi_0' \in \Pi_0 \right\}.
\end{align*}
The following lemma gathers some basic closure properties about $\mathcal{W}$.

\begin{lemma} \label{lemma:closure}
    The following are true:
    \begin{enumerate}
        \item $w \in \mathcal{W}$ implies $1-w \in \mathcal{W}$.
        \item $w, w' \in \mathcal{W}$ implies $ww' \in \mathcal{W}$.
        \item $\mathcal{W}$ is convex.
        \item If $u,v,w \in \mathcal{W}$, then $wu + (1-w) v \in \mathcal{W}$.
        \item If $w_n \in \mathcal{W}$ and $w_n \to w$ pointwise on $\Phi$, then $w \in \mathcal{W}$.
    \end{enumerate}
\end{lemma}
\begin{proof}
    The first result follows from observing that $\pi_0 \oplus_{1-w} \pi_0' = \pi_0' \oplus_{w} \pi_0$. Next, observe 
    \begin{align*}
    (\pi_0 \oplus_w \pi_0') \oplus_{w'} \pi_0'
        & = w'(\pi_0 \oplus_w \pi_0') + (1-w')\pi_0' \\
        & = w'\left(w\pi_0+(1-w)\pi_0'\right)+(1-w')\pi_0' \\
        & = ww'\pi_0+\left(w'(1-w)+1-w'\right)\pi_0' \\
        & = ww'\pi_0+(1-ww')\pi_0' \\
        & = \pi_0 \oplus_{ww'} \pi_0',
\end{align*}
    so we see that  $\pi_0 \oplus_{w w'} \pi_0' =  ( \pi_0 \oplus_w \pi'_0) \oplus_{w'} \pi'_0$. The second result then follows because  $\pi_0 \oplus_w \pi'_0 \in \Pi_0$. Convexity of $\mathcal{W}$ holds because
    \begin{align*}
       \pi_0 \oplus_{\lambda w + (1-\lambda) w'} \pi'_0  & =  \big( (\lambda w + (1-\lambda) w') \circ \varphi \big) \pi_0 + \big( \big( 1 - (\lambda w + (1-\lambda) w') \big) \circ \varphi \big) \pi'_0 \\ 
       & = \lambda \big( ( w \circ \varphi ) \pi_0 +  (1 - (w \circ \varphi)) \pi'_0  \big)  + (1-\lambda)  \big( ( w' \circ \varphi ) \pi_0 +  (1 - (w' \circ \varphi)) \pi'_0  \big) \\
       & = \lambda (\pi_0 \oplus_w \pi'_0) + (1-\lambda) (\pi_0 \oplus_{w'} \pi'_0), 
    \end{align*}
    and $\lambda (\pi_0 \oplus_w \pi'_0) + (1-\lambda) (\pi_0 \oplus_{w'} \pi'_0) \in \Pi_0$ by convexity of $\Pi_0$. 

    As for the fourth property, observe that for any  $u,v,w \in \mathcal{W}$, 
    \begin{align*}
        (\pi_0 \oplus_u \pi_0') \oplus_w (\pi_0 \oplus_v \pi_0')
        & = w(\pi_0 \oplus_u \pi_0') + (1-w)(\pi_0 \oplus_v \pi_0') \\
        & = w\left(u\pi_0+(1-u)\pi_0'\right) + (1-w)\left(v\pi_0+(1-v)\pi_0'\right) \\
        & = \left(wu+(1-w)v\right)\pi_0 + \left(w(1-u)+(1-w)(1-v)\right)\pi_0' \\
        & = \left(wu+(1-w)v\right)\pi_0 + \left(1-wu-(1-w)v\right)\pi_0' \\
        & = \pi_0 \oplus_{wu+(1-w)v} \pi_0'.
    \end{align*}
    Since $\pi_0 \oplus_u \pi_0', \pi_0 \oplus_v \pi_0' \in \Pi_0$, it follows that $wu+(1-w)v \in \mathcal{W}$.   Finally,  $\mathcal{W}$ inherits weak closedness from $\Pi_0$ because $\int f \, d(\pi \oplus_{w_n} \pi') \to \int f \, d(\pi \oplus_{w} \pi')$ for any $f \in C_b(\Theta)$ by dominated convergence.
\end{proof}

Define the set  $\Sigma \equiv \left\{B \in \mathcal{B}_\Phi : \1_B \in \mathcal{W} \right\}$. We know from taking $A = Y$ in  Lemma~\ref{lemma:likelihood} that $1 \in \mathcal{W}$, which implies $\Phi \in \Sigma$. The first and second properties in Lemma~\ref{lemma:closure} imply that $\Sigma$ is closed under complements and finite intersections. Weak closedness then implies $\Sigma$ is closed under countable intersections, allowing us to conclude that $\Sigma$ is a $\sigma$-algebra. The following lemma establishes that $\Sigma$ is the collection of all Borel sets $\mathcal{B}_\Phi$.

\begin{lemma} \label{lemma:sigma}
 For any $B \in \mathcal{B}_\Phi$, we have $\1_B \in \mathcal{W}$. That is, $\Sigma  = \mathcal{B}_\Phi$.
\end{lemma}
\begin{proof}
    Note that the desired result holds if  $\{\phi \in \Phi : \ell(A \mid \phi)>c\} \in \Sigma$ for every $A \in \mathcal{Y}$ and $c \in [0,1)$ since  $\Phi \subseteq \Delta(Y)$ and $Y$ is Polish \citep[e.g., Theorem 19.7,][]{Aliprantis_Border_2006}. Recall from Lemma~\ref{lemma:likelihood} that  $w_A(\cdot) = \ell(A \mid \cdot) \in \mathcal{W}$ for any $A \in \mathcal{Y}$, in which case it  suffices to show that $\1_{ \bar w > c} \in \mathcal{W}$  for any $\bar w \in \mathcal{W}$ and $c \in [0,1)$. 
    
    We proceed via a standard polynomial approximation argument. Define $z \equiv \frac{\bar w+1-c}{2} \in \mathcal{W}$, and consider the polynomial 
    \begin{align*}
        F(x) = 3x^2 - 2x^3 = x\left( 1-(1-x)^2\right) + (1-x)x^2.
    \end{align*}
    Taking $w = z$, $u = 1- (1-z)^2$, and $v = z^2$, since we know that $w,u,v \in \mathcal{W}$, Lemma~\ref{lemma:closure} implies $F(z) \in \mathcal{W}$. Applying the same argument recursively,  we obtain $F^n(z) \in \mathcal{W}$ for every $n \geq 1$, where $F^n$ denotes the $n$-fold composition of $F$.

    It is easy to verify that  $F(x)$  has three fixed points $x \in \{0, \frac{1}{2}, 1\}$, where $\frac{1}{2}$ is unstable while $0,1$ are stable. Since  $F$ maps $[0,1]$ into itself, repeated iterations converge to one of the three fixed points, such that
    \begin{align*}
        \lim_{n\to\infty} F^n(x)
        =
        \begin{cases}
            0 & \text{if } x < 1/2,\\
            1/2 & \text{if } x = 1/2,\\
            1 & \text{if } x > 1/2.
        \end{cases}
    \end{align*}
    Since $z(\phi)>1/2$ if and only if $\bar w(\phi)>c$, weak closedness of $\mathcal{W}$ implies $\1_{\{\bar w>c\}} + \frac{1}{2} \1_{\{\bar w=c\}}  \in \mathcal{W}$. We know $\mathcal{W}$ is closed under multiplication from  Lemma~\ref{lemma:closure}, which implies $(\1_{\{\bar w>c\}} + \frac{1}{2} \1_{\{\bar w=c\}} )^n = \1_{\{\bar w>c\}}  + \frac{1}{2^n} \1_{\{\bar w=c\}} \in \mathcal{W}$ for any $n \geq 1$. Since  Lemma~\ref{lemma:closure} also tells us that $\mathcal{W}$ is closed under pointwise convergence, we can conclude that $\1_{\{\bar w>c\}} \in \mathcal{W}$, as desired. 
\end{proof}

For any $B \in \mathcal{B}_\Phi$, the above lemma and the definition of $\mathcal{W}$ implies  $\pi_0 \oplus_{\1_B} \pi'_0 \in \Pi_0$ for all $\pi_0, \pi_0' \in \Pi_0$.  Define $\rho_0 \equiv \pi_0 \circ \varphi^{-1}$ and $\rho'_0 \equiv \pi'_0 \circ \varphi^{-1}$ as the corresponding reduced-form priors, and observe that 
\begin{align*}
    1 & = (\pi_0 \oplus_{\1_B} \pi'_0)(\Theta) \\
    & = \int_\Theta \1_B(\varphi(\theta)) \, \pi_0(d\theta) + \int_\Theta \1_{B^c}(\varphi(\theta)) \, \pi'_0(d\theta) \\
    & = \rho_0(B) + \rho'_0(B^c).
\end{align*}
Since our choice of $B$ was arbitrary, we have $\rho_0 = \rho'_0$. That is, every $\pi_0 \in \Pi_0$ induces the same reduced form marginal $\rho_0 = \pi_0 \circ \varphi^{-1}$.

 Define the set of measurable probability kernels $q: \Phi \to \Delta(\Theta)$ induced by $\Pi_0$ as
\begin{align*}
    \mathcal{C} = \left\{ q : \Phi \to \Delta(\Theta) :  \int_\Phi q(\cdot \mid \phi) \, \rho_0(d\phi) \in \Pi_0  \,\, \text{ and }  \, \, q(IS_\Theta(\phi) \mid \phi) = 1, \,\,  \rho_0\text{-a.s.}  \right\}.
\end{align*}
These are the conditional distributions of $\pi_0 \in \Pi_0$ given the identified set $IS_\Theta(\phi)$.  We say that  a collection of probability kernels $\mathcal{C}$ is decomposable if for any $q,q' \in \mathcal{C}$ and $B \in \mathcal{B}_\Phi$, we have $\1_B q + \1_{B^c} q' \in \mathcal{C}$. 

\begin{lemma} \label{lemma:decompose}
     $\mathcal{C}$ is decomposable. 
\end{lemma}
\begin{proof}
    Suppose  $q,q' \in \mathcal{C}$ induce $\pi_0,\pi_0'$. Observe that for any $B \in \mathcal{B}_\Phi$,
\begin{align*}
    \int_\Phi \left( \1_B q + \1_{B^c} q' \right) \rho_0(d\phi)
    & = \int_\Phi \1_B(\phi) q(\cdot \mid \phi) \, \rho_0(d\phi) + \int_\Phi \1_{B^c}(\phi) q'(\cdot \mid \phi) \,  \rho_0(d\phi) \\
    & = (\1_B \circ \varphi) \int_\Phi q(\cdot \mid \phi) \, \rho_0(d\phi) + (\1_{B^c} \circ \varphi) \int_\Phi q'(\cdot \mid \phi) \, \rho_0(d\phi) \\
    & = (\1_B \circ \varphi)\pi_0 + (\1_{B^c} \circ \varphi)\pi_0'.
\end{align*}
The final term equals $\pi_0 \oplus_{\1_B} \pi'_0$, which we know by Lemma~\ref{lemma:sigma} is in $\Pi_0$. We thus have that $\1_B q + \1_{B^c}q' \in \mathcal{C}$, as desired.  
\end{proof}

Because every prior in $\Pi_0$ has the reduced-form marginal $\rho_0$, the existence of regular conditional distributions implies that $\mathcal{C}$ is non-empty.  It is also easy to see that $\mathcal{C}$ inherits the convexity of $\Pi_0$. Moreover, $\mathcal{C}$ is closed under convergence in $\rho_0$-probability by a standard bounded convergence argument.\footnote{Any convergent sequence $\{q^n\} \subseteq \mathcal{C}$ admits a $\rho_0$-a.s. convergent subsequence. Since the identified sets are closed, the limiting kernel remains supported on the corresponding identified sets. By the Bounded Convergence Theorem, the induced priors $\int_\Phi q^n(\cdot \mid \phi) \, \rho_0(d\phi)$ converge weakly, so closedness of $\Pi_0$ implies that the limiting kernel belongs to $\mathcal{C}$.} We therefore have that $\mathcal{C}$ is non-empty, closed, convex, and decomposable.

The remainder of the proof constructs a correspondence $\bP_\Phi:\Phi\rightrightarrows\Delta(\Theta)$ satisfying the conditions in Definition~\ref{def:PI}, such that $\Pi_0=\E_{\rho_0}[\bP_\Phi]$. We use a Castaing representation argument analogous to that of Theorem 2.1.10 of  \cite{Molchanov_2017} (see also Theorem 2.10 in \cite{Molchanov_Molinari_2018}). Since $\Phi$ is a separable metric, the space $L^0_{\rho_0}(\Delta(\Theta))$ of measurable maps from $\Phi$ to $\Delta(\Theta)$ (identified up to $\rho_0$-almost sure equality), is separable under convergence in $\rho_0$-probability.

Because $\mathcal C$ is a subset of the separable metric space $L^0_{\rho_0}(\Delta(\Theta))$, it admits a countable dense subset. By convexity of $\mathcal C$, we can take the rational convex hull of this dense subset to obtain a countable dense collection $\{q^n\}_{n\geq1} \subseteq \mathcal C$ that is closed under rational convex combinations. Define the correspondence
\begin{align*}
    \bP_\Phi(\phi) \equiv \overline{\{q^n(\cdot\mid\phi):n\geq1\}},
\end{align*}
where the closure is taken under weak convergence in $\Delta(\Theta)$. The correspondence is measurable because it is the pointwise closure of a countable family of measurable kernels. Moreover, $\bP_\Phi$ has closed and convex values because $\{q^n\}_{n\geq1}$ is closed under rational convex combinations. Finally, since each $q^n$ is supported on $IS_\Theta(\phi)$ $\rho_0$-almost surely and $IS_\Theta(\phi)$ is closed, we have $\bP_\Phi(\phi) \in \K^c(\Delta \Theta)$, $\rho_0$-almost surely. 

\begin{lemma}
    The measurable selections of $\bP_\Phi$ coincide with $\mathcal{C}$.
\end{lemma}
\begin{proof}
    Suppose $q$ is a measurable selection of $\bP_\Phi$.  Recall that $\Delta(\Theta)$ is a separable metric under weak convergence. Let $d$ denote a metric that induces this topology, and define
    \begin{align*}
        n_k (\phi) \equiv \min \left\{ n \geq 1 : d( q(\cdot \mid \phi) , q^n(\cdot \mid \phi)) < \frac{1}{k} \right\}, \qquad k \geq 1,
    \end{align*}
    which is measurable and finite almost surely. 
    
    Recall from Lemma~\ref{lemma:decompose} that $\mathcal{C}$ is decomposable. For each $k,m \geq 1$, define
    \begin{align*}
        q^{k,m} (\cdot \mid \phi) \equiv \sum_{n=1}^m \1_{ \{n_k (\phi) = n\} } q^n(\cdot \mid \phi) + \1_{ \{n_k(\phi) > m\} } q^1(\cdot \mid \phi).
    \end{align*}
    Since $\{n_k(\phi) = n\}$ for $n = 1 , \ldots , m$ and $\{n_k (\phi) > m\}$ form a finite measurable partition of $\Phi$, repeated applications of decomposability implies that $q^{k,m} \in \mathcal{C}$. Moreover, we have 
    \begin{align*}
        \lim_{m \to \infty} q^{k,m} (\cdot \mid \phi) = q^{n_k(\phi)}(\cdot \mid \phi), \qquad \rho_0\text{-a.s.}
    \end{align*}
    because $n_k(\phi)$ is finite $\rho_0$-almost surely. Closedness of $\mathcal{C}$ then  implies $q^{n_k(\phi)}(\cdot \mid \phi) \in \mathcal{C}$. Since 
    \begin{align*}
        d \left( q^{n_k(\phi)} (\cdot \mid \phi), q(\cdot \mid \phi) \right) < \frac{1}{k}, \qquad \rho_0\text{-a.s.},
    \end{align*}
     we have $q^{n_k(\phi)} \to q$ in $\rho_0$-probability as $k \to \infty$, and closedness of $\mathcal{C}$ implies $q \in \mathcal{C}$.

       Conversely, fix some $q \in \mathcal{C}$. By density of $\{q^n\}$, there exists some sequence $\{q^{n_s}\}_{s\geq 1}$ that converges to $q$ in $\rho_0$-probability. There then exists a subsequence that almost surely converges to $q$, such that $q(\cdot \mid \phi) \in \bP_\Phi(\phi)$, $\rho_0$-almost surely. Hence, $q$ is a measurable selection of $\bP_\Phi$, as desired.  
\end{proof}

By definition of $\mathcal C$, integrating any kernel in $\mathcal C$ against $\rho_0$ yields a prior in $\Pi_0$. Conversely, since every prior in $\Pi_0$ has reduced-form marginal $\rho_0$, the existence of regular conditional distributions implies that every prior in $\Pi_0$ can be obtained by integrating some kernel in $\mathcal C$ against $\rho_0$. We then have 
\begin{align*}
     \Pi_0 &  =  \left\{ \int_\Phi q(\cdot \mid \phi) \, \rho_0(d\phi) : q \in \mathcal{C} \right\}    = \E_{\rho_0}\left[ \bP_\Phi \right],
\end{align*}
where the last equality holds by definition of Aumann expectations.  \hfill $\blacksquare$

\subsection{Proof of Proposition~\ref{prop:Blackwell}} \label{pf:Blackwell}
 
Fix the reduced-form parameter space $\Phi$. The fact that the inequalities in the latter two statements hold with equality for singleton $\mathcal{A}$ follows from Theorem~\ref{thm:two_ways}, since $\ell_1, \ell_2$ set-identify $\theta$ up to a coarsening of $\Phi$ and $\Pi_0$ is partially identified by $(\rho_0, \Phi)$.
  
($1. \Leftrightarrow 2.)$  Suppose the prior set $\Pi_0$ is partially identified by $(\rho_0,\Phi)$ and   consider experiments that set-identify $\theta$ up to a coarsening of $\Phi$. By Lemma~\ref{lemma:updating_equivalence}, the posterior sets satisfy $\Pi_y^\ell = \E_{\rho_y^\ell}[\bP_\Phi]$, where $\rho_y^\ell$ is the reduced-form posterior induced by $\rho_0$ and $\ell$. Define the worst-case conditional expected loss
\begin{align*}
    \overline{L}(a,\phi)\equiv\sup_{\pi_\phi \in \Pi_\phi} \E_{\pi_\phi} [L(\theta,a) ],
\end{align*}
where $\Pi_\phi = \bP_\Phi(\phi)$ for all $\phi \in \Phi$. 

By the Aumann expectation representation of $\Pi_y^\ell$, we can rewrite the value function as
\begin{align*}
    V_{\Pi_0,\mathcal A,L}(\ell) & = -\E_{\mu^\ell}\left[\min_{a\in\mathcal A}\int_\Phi\overline L(a,\phi)\,\rho_Y^\ell(d\phi)\right],
\end{align*}
which coincides with the value function $V_{\{\rho_0\}, \mathcal{A}, \overline{L}}$ in the standard Bayesian setting with action set $\mathcal{A}$, state space $\Phi$, prior $\rho_0 \in \Delta(\Phi)$, and loss $\overline{L}(a,\phi)$.   The desired equivalence then follows from Blackwell's original theorem, which says  $\ell_1$ is more valuable than $\ell_2$ if and only if $V_{\{\rho_0\}, \mathcal{A}, \overline{L}}(\ell_1) \geq V_{\{\rho_0\}, \mathcal{A}, \overline{L}}(\ell_2)$ for all $(\mathcal{A}, \overline{L}, \rho_0)$.

$(1. \Leftrightarrow 3.)$ The proof proceeds as in the proof of $1 \Leftrightarrow 2$, but  now we define the worst-case conditional expected loss as
\begin{align*}
    \overline{L}(\sigma, \phi) \equiv \sup_{\pi_\phi \in \Pi_\phi} \E_\sigma\left[ \E_{\pi_\phi}[L(\theta ,a)] \right].
\end{align*}
The value function with randomized actions can then be rewritten as
\begin{align*}
    W_{\Pi_0, \mathcal{A}, L}(\ell) = - \E_{\mu^\ell} \left[ \inf_{\sigma \in \Delta(\mathcal{A})} \int_\Phi \overline{L}(\sigma, \phi)  \, \rho_Y^\ell(d\phi) \right].
\end{align*}
Because $\overline{L}(\cdot,\phi)$ is bounded over $\Delta( \mathcal{A})$ and Blackwell's original theorem extends to the case when the action set is a Polish space \citep[e.g.,][]{JET_blackwell_infinite}, the desired result follows. \hfill $\blacksquare$

\subsection{Proof of Lemma~\ref{lemma:preference_preservation}} \label{pf:preference_preservation}

Sufficiency of Aumann plausibility is immediate since it implies $\Pi_0^* = \Pi_0$. As for necessity, suppose the $\Pi_0$-minimax and $\Gamma^*(\Pi_0)$-minimax criteria induce identical preferences for every loss $L$ and finite action space $\mathcal{A}$. Recall that we denote the support function corresponding to each $\Pi \in \K^c(\Delta \Theta)$ by $h_\Pi(f) = \sup_{\pi \in \Pi} \int_\Theta f \, d\pi$ for any $f \in C_b(\Theta)$. Since $\Pi_0^*$ is the marginalization of $\Gamma^*(\Pi_0) \supseteq \Gamma(\Pi_0)$, we know that $\Pi_0^* \supseteq \Pi_0$, and thus $h_{\Pi_0}(\cdot) \leq h_{\Pi_0^*}(\cdot)$.  If this inequality always holds with equality, the Hahn-Banach separation theorem implies  $\Pi_0 = \Pi_0^*$, so $(\Pi_0, \ell)$ satisfies Aumann plausibility.

Now suppose, for a contradiction, there exists some $f \in C_b(\Theta)$ at which $h_{\Pi_0} (f) < h_{\Pi_0^*} (f)$. Choose some $c \in \R$ satisfying $h_{\Pi_0} (f) < c < h_{\Pi_0^*} (f)$. Consider the action space $\mathcal{A} = \{a,b\}$ and loss $L(\theta,a) = f(\theta)$ and $L(\cdot, b) = c$. Under $\Pi_0$-minimax, we have
\begin{align*}
    \sup_{\pi_0 \in \Pi_0} \E_{\pi_0} \left[ L(\theta, a) \right] = h_{\Pi_0}(f) < c =  \sup_{\pi_0 \in \Pi_0} \E_{\pi_0} \left[ L(\theta, b) \right],
\end{align*}
so action $a$ is preferred to action $b$. Under  $\Gamma^*(\Pi_0)$-minimax, we have
\begin{align*}
    \sup_{\gamma^* \in \Gamma^*(\Pi_0)} \E_{\gamma^*} \left[ L(\theta, a) \right] = h_{\Pi_0^*}(f) > c =   \sup_{\gamma^* \in  \Gamma^*(\Pi_0)} \E_{\gamma^*} \left[ L(\theta, b) \right],
\end{align*}
so  action $b$ is preferred to action $a$. This yields a contradiction, since we have by assumption that  the $\Pi_0$-minimax and $\Gamma^*(\Pi_0)$-minimax criteria induce identical preferences over $\mathcal{A}$. \hfill $\blacksquare$

\bibliography{ref}

\end{document}